\documentclass[journal]{IEEEtran}
\usepackage{graphicx} 
\usepackage{amsmath,amssymb,amsbsy,latexsym}
\usepackage[caption=false,font=footnotesize]{subfig}
\usepackage{url}
\usepackage{algorithm}
\usepackage{algpseudocode}

\usepackage{amsmath, amssymb, amsthm}
\usepackage[caption=false,font=footnotesize]{subfig}
\usepackage{array}
\usepackage{makecell}
\usepackage{graphicx}
\usepackage[mathscr]{eucal}
\usepackage{cite}
\usepackage{enumerate}
\usepackage[table]{xcolor}
\usepackage{empheq}  
\usepackage{enumitem}
\setlist{nolistsep}
\IEEEoverridecommandlockouts                              

\newcolumntype{x}[1]{>{\centering\arraybackslash}p{#1}}
\DeclareMathOperator{\Tr}{Tr}
\newcommand{\FGLS}{\text{FGLS}}
\newcommand{\R}{\mathbb{R}}
\newcommand{\argmin}{\arg\!\min}
\newcommand{\vep}{\varepsilon}
\newcommand{\rar}{\rightarrow}
\newcommand{\mb}{\mathbb}
\newcommand{\mc}{\mathcal}

\newcommand{\bmat}[1]{\begin{bmatrix}#1\end{bmatrix}}

\newtheorem{proposition}{Proposition}

\newtheorem{definition}{Definition}

\newif\ifplot
\plottrue
\newcommand{\LB}{\text{LB}}
\newcommand{\np}{p}
\newcommand{\ms}{m_{\text{stop}}}
\newcommand{\AIC}{\text{AIC}}
\newcommand{\cFGLS}{\text{cFGLS}}
\definecolor{CornflowerBlue}{rgb}{0.258824,0.258824,0.435294}
\definecolor{SkyBlue}{rgb}{0.196078,0.6,0.8}
\definecolor{dblue}{rgb}{.098,.243,.424}
\definecolor{lblue}{rgb}{.33,.57,.835}
\definecolor{llblue}{rgb}{.447,.643,.831}
\definecolor{mblue}{rgb}{0.176, 0.380, 0.659}
\definecolor{lcomp}{rgb}{.969,.765,.416}
\definecolor{ddorange}{rgb}{0.624, 0.365, 0}
\definecolor{dorange}{rgb}{0.72, 0.506, 0.125}
\definecolor{lorange}{rgb}{0.961, 0.678, 0.165}
\definecolor{lgreen}{rgb}{.812,.969,.435}
\definecolor{lyellow}{rgb}{1,.859,.451}
\definecolor{dyellow}{rgb}{.651,.482,0}
\definecolor{lrred}{rgb}{1,.6,.451}
\definecolor{dred}{rgb}{.85,.176,0}
\colorlet{mmblue}{blue!50}
\newcommand{\lightercolor}[3]{
    \colorlet{#3}{#1!#2!white}
}
\newcommand{\darkercolor}[3]{
    \colorlet{#3}{#1!#2!black}
}
\lightercolor{mmblue}{50}{llblue}
\darkercolor{mmblue}{50}{ddblue}
  \colorlet{dgreen}{green!70!black}

\newcommand{\players}{players }

\newcommand{\player}{player }
\newcommand{\la}{\langle}
\newcommand{\ra}{\rangle}
{\theoremstyle{plain}\newtheorem*{example*}{Example}}
\begin{document}

\title{\LARGE \bf
A Robust Utility Learning Framework via Inverse Optimization}

\author{\IEEEauthorblockN{Ioannis C.
Konstantakopoulos\thanks{\textsuperscript{*}Authors
contributed equally}\textsuperscript{*}, Lillian J. Ratliff\textsuperscript{*}, Ming Jin, S.
Shankar Sastry, Costas J. Spanos}
\thanks{I. Konstantakopoulos, M. Jin, S. Sastry, and C. Spanos are with the
Electrical Engineering and Computer Sciences Department, University of
California, Berkeley, Berkeley, CA 94720. email: {\tt $\{$ioanniskon, jinming,
spanos, sastry$\}$@eecs.berkeley.edu} }
\thanks{L. Ratliff is with the Electrical Engineering Department, University of
Washington, Seattle, WA 98195. email: {\tt ratliffl@uw.edu}}
}

\maketitle

\thispagestyle{empty}
\pagestyle{empty}

\begin{abstract}
   In many smart infrastructure applications,  flexibility in
    achieving
sustainability goals can be gained by engaging end-users.
However, these users often have heterogeneous preferences that are unknown to
the decision-maker tasked with improving operational efficiency. Modeling user
interaction as a
continuous game between non--cooperative players, we propose a robust 
parametric utility learning framework that employs constrained feasible generalized least squares estimation
with heteroskedastic
inference. 
To improve forecasting performance, we extend the robust
utility learning scheme by employing bootstrapping with bagging, bumping, and
gradient boosting ensemble methods. 
Moreover, we estimate the noise covariance which provides approximated
correlations between \players
which we leverage to develop a novel \textit{correlated utility learning} framework.
We apply the proposed methods both to a toy example arising from Bertrand-Nash
competition between two firms 
as well as to data from a social game experiment designed to encourage
energy efficient behavior amongst smart building occupants. Using occupant
voting data for shared resources such as lighting, we simulate the game defined by the estimated
utility functions to  
demonstrate the performance of the proposed methods.

\end{abstract}


\section{Introduction}
\label{sec:intro}
Due to pervasive utilization of Internet of Things and Cyber-Physical
Systems sensing/actuating platforms, we are increasingly observing human
decision-makers being integrated into operational and
managerial decisions in  infrastructure systems. Their actions can be
leveraged to increase both resilience and sustainability thereby making smart
infrastructure a worthwhile investment.   Smart buildings,
being no exception, are a fundamental component of emerging \emph{smart cities};
their efficient design and operation enables flexibility---e.g., by automatically
shifting or curtailing demand during peak hours---in
making urban spaces sustainable.
More abstractly, in many infrastructure systems there is often an entity acting as
a \emph{planner} (e.g., facility managers, departments of transportation, etc.) that
 introduces incentives or control policies to coordinate
 autonomously acting agents in the system (e.g., selfish human decision-makers)
 so that their collective behavior leads to system-level efficiency gains.

 %
 One approach to designing such policies is to leverage game-theoretic models of
 decision-making in an optimization framework to produce policies that encourage
 or induce behavior that optimizes an objective~\cite{bolton:2005aa,laffont:2002aa}. 
 Often the planner has at best 
 a prior on the
 decision-making model of the individual agents. Such information asymmetries
lead to inefficiencies~\cite{bolton:2005aa, ratliff:2015aa}.  In this paper, we propose a framework for
 estimating decision-making models of self-interested decision-makers
 consuming a shared resource (e.g., lighting in a smart building) that can be
 leveraged in control or incentive design to aid in closing the efficiency gap.

%
To concretize ideas, consider a smart building---an example we will return
to throughout the text. A facilities manager may be incentivized or even tasked to encourage energy
efficient behavior if they are accountable for energy costs or  are required,
e.g., to maintain
an operational excellence measure~(see, e.g.,~\cite{OpEx:2015aa,aswani:2012aa}).
At the same time, the facilities manager 
generally must also
ensure user comfort and productivity~\cite{jin2016occupancy}. Beyond these motivations, demand response
(DR) programs are being rolled out by utility
companies and third-party solution providers with the goal of correcting for
improper load forecasting.
Participating consumers
 decide to change their consumption when DR events are called \cite{jin2017mod}. The facilities manager may
 be required to keep this schedule. 
 
 Smart building technologies enable new
 avenues for facilities managers to keep such a prescribed schedule 
 via automation or integrating
the end-user. 
%
Yet,  in office buildings the occupants, as employees, typically are not
 responsible for paying for the energy resources they consume.
 Hence, there is often a misalignment between the incentives of the
 facilities manager and the occupants. Social games are a means to engage the
 occupants to address these inefficiencies. In Section~\ref{sec:apptosoc}, we
 describe one such social
 game that we designed and implemented on the UC Berkeley campus,
 aimed at incentivizing energy
 efficient consumption of shared resources by leveraging building automation.


%
The broader purpose of this paper is to present a general framework that
leverages game-theoretic concepts to learn models of players'
decision-making in competitive environments such as the building
energy social game described above. The framework supports learning 
 agents' preferences over shared resources as well as understanding how
preferences change as a function of external stimuli such as
 physical control or  incentives. 
Such a framework can be used in the design of incentive mechanisms that realign agents'
preferences with those of the planner---which often represent system-level
performance criteria---through fair
compensation. 


%
More concretely, we model decision-making agents as \emph{utility maximizers} and, using inverse optimization and game-theoretic techniques, we
derive a robust scheme to infer their utility functions.
At the core of our approach is the fact that we
model the agents as non-cooperative \players in a game playing according to a \emph{Nash
equilibrium strategy}. From this point of view, agents are strategic entities
that make decisions based on their own preferences despite 
others. The game-theoretic framework both allows for qualitative insights to be
made about the outcome of such selfish behavior---more so than a simple
prescriptive model---and, more importantly, can be
leveraged in designing mechanisms for incentivizing agents.

We assume a parametric form of utility function
for each \player that is dependent on the decisions of others. 
Correlations between players' decisions are not known \emph{a priori}.  
Assuming observations are approximately Nash equilibria, we use first--
and second--order conditions on player utility functions to construct a
constrained regression model.  
The result is as a constrained Generalized Least Squares
(cGLS) problem with non-spherical noise error
terms.  
Using constrained Feasible Generalized Least Squares (cFGLS), an implementable
version of cGLS, we utilize  heteroskedastic
inference to approximate the correlated
errors.

Noting  that data sets of observed decisions often
    may be small relative to the number of model parameters in practice, we
    employ 
bootstrapping to generate pseudo-data from which we learn
additional estimators. The bootstrapping process allows us to derive an asymptotic approximation of the bias and
standard error of an estimator.
We utilize ensemble methods such as bagging, bumping, and gradient boosting
to extract an estimator from the pseudo-data generated estimators that results
in a reduced forecasting error. The ensemble methods are robust under noise and autocorrelated error terms. 
We apply the robust utility learning framework to a model of
Bertrand-Nash competition between firms in order to illustrate the framework and
its performance. 

Building on the robust utility learning framework, we use the approximated standard
error to derive an  innovative utility learning method in which we
modify
players' utility functions to create a \emph{correlated game}. The resulting correlated
utility learning method leverages
correlations between \players and the ensemble estimators to
minimize the estimation error by optimizing scaling coefficients that appear in
the correlated game utility functions.
Applying this method results in a significant
improvement over the constrained Ordinary Least Squares (cOLS) estimations and outperforms many of the ensemble
methods. It also provides insights into how \players interact with one another
and indicates which \players are potentially forming coalitions. Moreover, this technique is amenable to online implementation
after an initial training period so that by using
cOLS estimators in the
correlated utility learning framework, our adaptive incentive design schemes,
introduced in~\cite{ratliff:2014aa,ratliff:2015aa}, can 
be made robust.

To demonstrate the efficacy of both the robust and correlated utility
learning frameworks, we apply them 
to data generated from the smart building social
game experiment we conducted. We show that estimating the players' utility 
functions via the proposed methods
results in a predictive model that outperforms several other standard
techniques such as Ordinary Least Squares (OLS).

The rest of the paper is organized as follows. We describe the abstracted game
framework for modeling the interaction of agents as well as define equilibrium
concepts in Section~\ref{sec:games}. In Section~\ref{sec:robutility}, we
formulate the robust utility learning framework and provide an algorithm for
implementing it. Section~\ref{sec:bertnash} contains the Bertrand-Nash
competition example and we present the correlated utility learning framework in
Section~\ref{sec:corutility}. 
In Section~\ref{sec:apptosoc}, we describe the social game experimental setup on the UC Berkeley campus within the CREST center\footnote{{\tt
http://crest.berkeley.edu/}},
provide a brief literature review, and present the results of both proposed
utility learning methods applied
to data from the social game. We make concluding remarks and discuss
future directions in Section~\ref{sec:conclusion}.
 

%

\section{Game Framework}
\label{sec:games}
In this section, we abstract the agents' decision-making processes in a
game--theoretic framework.

%

%

\subsection{Agent Decision-Making Model}
Consider $\np$ agents\footnote{We refer to the decision-makers
    as \emph{agents} and use the term interchangeably with \emph{players}.}---i.e.~decision-making entities---indexed by the set
$\mc{I}=\{1,\ldots, \np\}$.
Each agent is modeled as a \emph{utility maximizer} that seeks to select $x_i\in
\mb{R}$ by optimizing 
\begin{equation}
  f_i(x_i,x_{-i})=f_i^{\text{nom}}(x_i,x_{-i})+f_i^{\text{inc}}(x_i, x_{-i}).
  \label{eq:utilnominc}
\end{equation}
where $f_i^{\text{nom}}(x_i,x_{-i})$ and $f_i^{\text{inc}}(x_i, x_{-i})$ are the
nominal and incentive components, respectively, of agent $i$'s utility function
and where
$x_{-i}=(x_1, \ldots, x_{i-1}, x_{i+1}, \ldots, x_n)\in \mb{R}^{n-1}$ is the
collective choices of all  agents excluding the $i$--th agent\footnote{Note that while for
notational simplicity we assume that $x_i\in \mb{R}$, the work easily extends to
 a higher
dimensional choice vector for each agent.}. 

The choice $x_i$ abstracts the agent's decision; it could represent, e.g.,
how much of a particular resource they choose to consume.
The nominal component of $f_i$ captures the agent's individual preferences over
$x_i$ and may depend on the decisions of others $x_{-i}$. The incentive component models the
portion of the agent's utility that can be designed by the planner; it also may
depend on the decisions of other agents. 

Agent $i$'s optimization problem is also subject to
constraints; the constraint set is given by $\mc{C}_i=\{x_i|\ h_{i,j}(x_i)\geq 0, j=1,\ldots, \ell_i\}$
where each $h_{i,j}$ is assumed to be a concave function of $x_i$. Such
constraints may encode cyber or physical constraints arising from the underlying
system---in the social game example presented in Section~\ref{sec:rusocialgame},
we will see that these constraints are physical bounds.  Thus,
given  $x_{-i}$, agent $i$ faces the following optimization problem:
\begin{equation}
  \max\{f_i(x_i,x_{-i})|\ x_i\in \mc{C}_i\}.
  \label{eq:opt-new}
\end{equation}

\subsection{Game Formulation}
The game $(f_1,\ldots, f_{\np})$ is a continuous game on a convex strategy space
$\mc{C}=\mc{C}_1\times\cdots\times \mc{C}_\np$.
To model the outcome of the strategic
interactions of agents, we use the \emph{Nash equilibrium} concept.
\begin{definition}
  A point $x\in \mc{C}$ is a {Nash equilibrium} for the game $(f_1, \ldots,
  f_\np)$ on $\mc{C}$
  if, for each $i\in \mc{I}$,
\begin{equation}
  f_i(x_i, x_{-i})\geq f_i(x_i',x_{-i})\ \ \forall \ x_i'\in \mc{C}_i.
  \label{eq:ineq}
\end{equation}
We say $x\in \mc{C}$ is an $\vep$--Nash equilibrium for $\vep>0$ if the above
inequality is relaxed:
\begin{equation}
  f_i(x_i, x_{-i})+\vep\geq f_i(x_i',x_{-i})\ \ \forall \ x_i'\in \mc{C}_i.
  \label{eq:ineq}
\end{equation}
\end{definition}
We say a point is a \emph{local Nash equilibrium} (respectively, a
\emph{$\vep$--local Nash equilibrium}) if there exists $W_i\subset \mc{C}_i$ such
that $x_i\in W_i$ and the above inequalities hold for all $x_i'\in W_i$.

If each $f_i$ is concave in $x_i$ and $\mc{C}$ is convex, then the game is a
$\np$--person concave game. In the seminal work by
Rosen~\cite{Rosen:1965tg}, it was shown that a (pure) Nash equilibrium exists for
every
concave game.

The Lagrangian of agent $i$'s optimization problem is given by
\begin{equation}
 \textstyle L_i(x_i, x_{-i}, \mu_i)=f_i(x_i, x_{-i})+\sum_{j\in
 \mc{A}_i(x_i)}\mu_{i,j}h_{i,j}(x_i)
  \label{eq:lagrangian}
\end{equation}
where $\mc{A}_i(x_i)$ is the active constraint set at $x_i$ and $\mu=(\mu_1,
\ldots, \mu_\np)$ with $\mu_i=(\mu_{i,j})_{j=1}^{\ell_i}$ are the Lagrange
multipliers. 
Assuming appropriate smoothness conditions on each $f_i$ and $h_{i,j}$, the differential game form~\cite{ratliff:2015aa},\cite{ratliff:2016aa}---which
characterizes
the first--order conditions of the game---is given
by
\begin{equation}
  \omega(x, \mu)=[D_1L_1(x, \mu_1)^\top\ \cdots\ D_\np L_\np(x, \mu_\np)^\top]^\top 
  \label{eq:omega}
\end{equation}
where $D_iL_i$ denotes the derivative of $L_i$ with respect to $x_i$.

Consider agent $i$'s optimization problem \eqref{eq:opt-new} with $x_{-i}$
fixed and where each $f_i$ and $h_{i,j}$ for $j\in \{1,\ldots, \ell_i\}$,
$i\in\mc{I}$ are concave, twice continuously differentiable functions. Then, assuming an appropriate constraint qualification
condition~\cite{Bertsekas:1999fk}, the necessary and sufficient conditions for
optimality of a point $x_i$ are as follows: there exists
$\mu_i\in\mb{R}^{\ell_i}_{+}$
such that (i) $D_iL_i(x, \mu_i)=0$; (ii) $\mu_i
h_{i,j}(x_i)=0$ for each $j\in\{1,\ldots, \ell_i\}$; (iii)
$h_{i,j}(x_i)\geq 0$ for each $j\in\{1, \ldots, \ell_i\}$. 
Regardless of the concavity assumption, the point $x_i$ is a local maximizer if $\mu_{i,j}>0$ and $z^\top D_{ii}^2L_i(x,
\mu_i) z<0$ for all $z\neq 0$ such that $D_ih_{i,j}(x_i)^\top z=0$ for
$j\in A_i(x_i)$. 
Such conditions motivate the following definition.


\begin{definition}[Differential Nash
        Equilibrium]
    Consider a game $(f_1,\ldots, f_{\np})$ on $\mc{C}$ where  
$f_i$ and $h_{i,j}$ for each $j\in\{1,\ldots, \ell_i\}$ and $i\in\mc{I}$ are
twice continuously differentiable. 
A point $x\in\mc{C}\subset\mb{R}^{\np}$ is a
differential Nash equilibrium if there is a $\mu\in
\mb{R}^{\sum_{i=1}^{\np}\ell_i}$ such that the pair $(x, \mu)$
satisfies
(i) $\omega(x, \mu)=0$; (ii) for each $i\in\mc{I}$,  
$z^\top D_{ii}L_i(x, \mu_i)z<0$ for all $z\neq 0$ such that
  $D_ih_{i,j}(x_i)^\top z=0$, and $\mu_{i,j}> 0$ for $j\in
  A_i(x_i)$.
 If, for a given $\vep>0$, (i')  $\omega(x,
  \mu)=\vep$ with all the other conditions being satisfied, then $x$ is a
  $\vep$--differential Nash equilibrium.
  \label{def:diffnash}
 \end{definition}
The above definition extends the definition of a differential Nash (if we
restrict to Euclidean spaces), first
appearing in~\cite{ratliff:2016aa}, to constrained games on Euclidean spaces. Using this definition, we can also extend
Proposition~1 of \cite{ratliff:2016aa}, again where strategy spaces are
restricted to be subsets Euclidean.

 \begin{proposition}
   A differential Nash equilibrium of the $\np$--person concave game $(f_1, \ldots,
   f_\np)$ on $\mc{C}$ is a local Nash equilibrium. 
   \label{prop:suffcond}
 \end{proposition}

 The proof is straightforward and we leave it to Appendix~\ref{app:proofs}. The
 proposition says that the conditions of Definition~\ref{def:diffnash} are
 sufficient for a local Nash.
In contrast to single-agent optimization problems, for games, the second
order conditions do not imply the equilibrium is isolated~\cite{ratliff:2016aa}. A sufficient condition guaranteeing that a Nash equilibrium $x$ is isolated is
that the Jacobian of $\omega(x, \mu)$, denoted
$D\omega(x, \mu)$, is invertible~\cite{ratliff:2015aa}. 

We use (necessary and sufficient) optimality conditions on individual
player optimization problems holding other players' strategies fixed  to formulate the utility learning framework.

\section{Robust Utility Learning}
\label{sec:robutility}

In previous work, we have explored utility learning and incentive design as a
coupled problem both in theory~\cite{ratliff:2015aa} and in
practice~\cite{ratliff:2014aa, ratliff:2014ac,jin2015rest}.
In the present work, we re-examine the utility learning
problem using statistical methods that serve to improve estimation  and prediction
accuracy.

%
Looking forward, our aim is to 
fold the new estimation scheme into the overall incentive design framework. This
goal motivates why we are interested in learning more than a simple predictive
model for agents, but rather a utility-based forecasting framework that accounts
for individual preferences. 

We parameterize  $f_i$
by
$\theta_i=(\theta_{i1}, \ldots,
\theta_{im_i})\in \mb{R}^{m_i}$ and 
 a finite set of basis functions $\{\phi_{ij}(x_i, x_{-i})\}_{j=1}^{m_i}$ such
 that
 \begin{equation}
     f_i(x; \theta_i)=\la \phi_i(x_i,x_{-i}), \theta_i\ra+\bar{f}_i(x)
     \label{eq:newparams}
 \end{equation}
 where $\phi_i=[\phi_{i,1}\ \cdots \ \phi_{i,m_i}]^\top$ and $\bar{f}_i(x)$ is a
 function that captures \emph{a priori} knowledge of the agent's utility
 function (e.g., the incentive component designed by the planner).


\subsection{Base Utility Estimation Framework}
\label{sec:baseutility}
We start by describing the basic utility estimation framework
using equilibrium conditions for the game played between the players.  The
utility learning framework we propose is quite broad in that it encompasses a
wide class of 
 continuous games.
 In previous works~\cite{ratliff:2014ac,ratliff:2015aa,jin2015rest} we have shown that the 
utility learning problem can be formulated as a convex optimization problem by
using first-- and second--order
 conditions for Nash equilibria. Let us briefly review this formulation as it
 serves as the basis for the robust utility learning method.

Each observation $x^{(k)}$ is assumed to be an $\vep$--approximate differential Nash
equilibrium where the superscript notation $(\cdot)^{(k)}$ indicates the $k$--th observation.
For each observation $x^{(k)}$, it may be the case that only a subset of the
players, say $\mc{S}^k\subset \mc{I}$ at observation $k$,
participate in the game.
Then notationally
each observation is such that
\begin{equation}
\textstyle x^{(k)}=\textstyle \left( x_j^{(k)} \right)_{j\in \mc{S}^k}.
  \label{eq:xkdef}
\end{equation}
If player $i$ participates in $n_i$ instances of the game, then there are $n_i$
observations for that player. Let $n=\sum_{i=1}^\np n_i$ be the total number of
observations.

We can consider first--order optimality conditions for each player's
optimization problem and define a residual function capturing the degree of
suboptimality of $x_i^{(k)}$~\cite{ratliff:2014aa},\cite{keshavarz:2011aa}.
Indeed, for player $i$'s
optimization problem,  
let the residual of the stationarity condition be given by 
 \begin{align}
  \label{eq:dstatzero}
     \textstyle  r_{\text{s},i}^{(k)}(\theta_i, \mu_i) &= \textstyle D_if_i(x_i^{(k)},
     x_{-i}^{(k)})+\sum_{j=1}^{\ell_i}\mu_i^jD_ih_{i,j}(x_i^{(k)})
\end{align}
and the residual of the complementary conditions be given by
\begin{align}
  r_{\text{c},i}^{j,(k)}(\mu) &=  \mu_i^jh_{i,j}(x_i^{(k)}), \ j\in \{1,\ldots,
  \ell_i\}.
\label{eq:dslaczero}
\end{align} 
Define \begin{equation}
  r_{\text{c},i}^{(k)}(\mu_i)=[r_{\text{c},i}^{1,(k)}(\mu_i)\ \cdots\ 
r_{\text{c},i}^{\ell_i,(k)}(\mu_i)].
  \label{eq:rcj}
\end{equation}
Using data from the players' decisions (e.g.,~lighting votes from the social game
experiment which we describe in Section~\ref{sec:experiment}), the base utility
learning framework consists of solving the optimization problem given by
\begin{equation}
 \begin{aligned}
   &   \min\limits_{\mu,
   \theta}\textstyle\sum_{i=1}^\np\sum_{k=1}^{n_i}\chi_i(r_{\text{s},i}^{(k)}(\theta, \mu),
  r_{\text{c},i}^{(k)}(\mu))\\
  &  \text{s.t.}\ \ \theta_i \in \Theta_i, \mu_i\geq 0\ \ \forall \
  i\in\{1, \ldots, \np\}
\end{aligned}
\tag{P}\label{opt-P}
\end{equation}
where $\Theta_i$ is a constraint set on the parameters $\theta_i$ that captures prior information
about the objective, 
 $\chi:\mb{R}^\np\times \mb{R}^{\sum_{i=1}^{\np}\ell_i}\rar \mb{R}_+$ is a non-negative, convex penalty function
satisfying $\chi(z_1,z_2)=0$ if and only if $z_1=0$ and $z_2=0$, i.e.~any norm
on $\mb{R}^\np\times \mb{R}^{\sum_{i=1}^\np\ell_i}$, and the inequality $\mu_i\geq 0$ is element-wise.

The goal of this optimization problem---which
is a finite dimensional optimization problem in the $\theta_i$'s---is to find $\theta_i$ for each
player such that $(\hat{f}_i)_{i\in\mc{I}}$ is
consistent (or approximately consistent) with the data.
As is noted in~\cite{keshavarz:2011aa}, we also remark that it is important
that the sets $\Theta_i$ contain enough prior information about the objectives
$f_i$ in order to prevent trivial solutions. For example, if it is the case that
$\bar{f}_i(x^{(k)})=0$ for each $k$
and each $\Theta_i=\mb{R}^{m_i}$ then the trivial solution
$\theta_i=\boldsymbol{0}_{m_i}$ is feasible. For many applications some \emph{a
priori} knowledge on part of the utility functions of players may be encoded in
each 
$\Theta_i$  (e.g., choosing
$\Theta_i$ such that $\theta_{1i}=1$ or similarly selecting the incentive
component of the utility, a design possibility for the
planner~\cite{ratliff:2015aa}) or through other normalization techniques to prevent such trivial solutions.
In the context of the social game application
(in Section~\ref{sec:rusocialgame}), we
explicitly discuss how to construct this constraint set in such a way that we ensure 
the estimated utility functions are concave which in turn guarantees that there
exists a Nash equilibrium to the estimated game.

\subsection{Robust Utility Learning}
\label{sec:robust}
Let us now formulate
a robust version of the utility learning framework that allows us to reduce our forecasting error and learn the
noise structure which can be leveraged in extracting \emph{pseudo--coalitions}
between players which we describe in the sequel.

Define  
\begin{equation}
  X_i^{(k)} = \bmat{D_i h_{i}(x_i^{(k)})&D_i
  \phi_i(x^{(k)}))\\ \hat{h}_{i}(x_i^{(k)})& \mathbf{0}_{\ell_i\times m_i} },
 \label{eq:regressormatrix}
\end{equation}
where
\begin{equation}
  \hat{h}_i(x_i)=\text{diag}(h_{i,1}(x_i), \ldots, h_{i,\ell_i}(x_i)),
  \label{eq:H}
\end{equation}
\begin{equation}
  D_ih_i(x_i)=[D_ih_{i,1}(x_i)\ \cdots\ D_ih_{i,\ell_i}(x_i)],
  \label{eq:DH}
\end{equation}
and $n_d=(\ell_i+1)n$ is the total number of data points. 
The regressor matrix is then defined as $X = \text{diag}(X_1, \cdots,
X_\np)\in \mb{R}^{n_d\times (\ell_i+1)\np}$ where $X_i=[ (X_i^{(1)})^\top \ \cdots \
(X_i^{(n_i)})^\top]^\top$.
Define the regression coefficient  
\begin{equation}
\beta =
[\mu_1^1 \ \ldots\ \mu_1^{\ell_1} \ \theta_1\ \cdots\ \mu_\np^1 \ \ldots \
\mu_\np^{\ell_\np} \ \theta_\np]^{\top} \in
\mb{R}^{(\ell_i+1)\np}
  \label{eq:reg-coeff}
\end{equation}
 and
the
observation matrix $Y = [ Y_1\ \cdots\ Y_\np ]^{\top}\in \mb{R}^{(\ell_i+1)\np}$ where 
\begin{equation}
    Y_i = [\bar{f}_i(x^{(1)}) \ \mathbf{0}_{\ell_i}  \
    \cdots \  \bar{f}_i(x^{(n_i)}) \ \mathbf{0}_{\ell_i}  ]^{\top}.
  \label{eq:obs-matrix}
\end{equation}

Using the Euclidean norm for $\chi$ in
\eqref{opt-P} leads to an OLS problem with inequality constraints---i.e. a
constrained OLS (cOLS):
\begin{equation}
\min\limits_{\beta}\left\{ \| Y - X \beta \|_{2}\big| \ \beta \in \mc{B}\right\}
\tag{P1}\label{opt-P1}
\end{equation}
where $\mc{B}=\{\beta|\ \theta_i\in \Theta_i, \mu_i\geq 0, \ \forall i\in
    \mc{I}\}$. Enforcing that each of the constraint sets $\Theta_i$ is encoded
    by inequalities on $\theta_i$, 
the above stated problem can be viewed as a classical multiple
linear regression model with inequality constraints described by the data
generation process
\begin{equation}
    Y = X\beta + \epsilon, \ \beta\in \mc{B} 
  \label{eq:mulregression}
\end{equation}
 where  
$\epsilon=(\epsilon_1, \ldots, \epsilon_\np)$ is the error term satisfying: 
    (i) $E(\epsilon | X) = 0^{n_d \times 1}$;
  (ii) $\text{cov}(\epsilon |
X) = \sigma^{2}I^{n_d \times n_d}$;
(iii)  $\{\epsilon_{i}\}_{i=1}^\np$ independent and identically distributed
(i.i.d) with a zero mean and $\sigma^{2}$ variance.
In addition, we assume $\epsilon$ is nonspherical~\cite{freedman2009statistical}.
With this general statistical model we are able to describe a data generation
processes in which the error terms are correlated or lack constant variance.
This fact will be leveraged in creating coalitions between players  as we
describe in Section~\ref{sec:corutility}.

Mathematically the nonspherical errors are modelled by
\begin{equation}
  \text{cov}(\epsilon |
X) = G \succ 0, \ G \in \mathbb{R}^{n_d \times n_d}.
  \label{eq:nonsphere}
\end{equation}
One drawback of this technique is that, given nonspherical standard errors, the
cOLS estimator is biased---that is,
it does not satisfy the Best Linear Unbiased Estimator (BLUE) property, a result
of the Gauss--Markov theorem~\cite[Theorem~1, Chapter~5]{freedman2009statistical}.
However, we can derive an unbiased estimator by 
multiplying \eqref{eq:mulregression} on the left with $G^{-\frac{1}{2}}$. This
leads to the cGLS statistical model
given by 
\begin{equation}
(G^{-\frac{1}{2}}Y) = (G^{-\frac{1}{2}}X)\beta +  (G^{-\frac{1}{2}}\epsilon), 
\ \beta \in \mc{B} 
  \label{eq:gls}
\end{equation}
which now satisfies the BLUE property.
In general, the explicit form of
$\text{cov}(\epsilon | X) = G$ is unknown.  
We use the
residuals~\eqref{eq:mulregression} to infer the noise by imposing structural
constraints on $G$.

We remark that there are many types of noise structures that can be
    used for imposing structure on $G$. We provide two example noise structures
that could be used. 
    The first is block diagonal structure~\cite[Chapter
    5]{freedman2009statistical}; in particular, we impose that 
$G  =\text{blkdiag}({K}_1, \cdots, {K}_\np)\in \R^{n_d \times n_d}$
where
${K}_i =\text{blkdiag}(B_{i,1}, \ldots, B_{i,n_i}) \in \R^{(\ell_i+1)n_i \times
(\ell_i+1)n_i}$
with each $B_{i,k}\in \mb{R}^{(\ell_i+1)\times (\ell_i+1)}$.
Estimating $\beta$ with cOLS, we
get $\hat{\beta}_{\text{cOLS}}$ with residual vector
$e=Y-X\hat{\beta}_{\text{cOLS}} \in \mb{R}^{(\ell_i+1) n}$. 
The residual vector $e$ can
be decomposed into residuals for each player by writing $e=[e_1^\top\ \cdots \
e_\np^{\top}]^\top$. We use $e_i$ to compute an estimate $\hat{K}_i$ of
${K}_i$ which is, in turn, used to compute $\hat{G}$.
 The residuals come in
triplets since for each $k$, $Y_i^{(k)}\in \mb{R}^{\ell_i+1}$. For ease of
presentation and comprehension, we will use a paired index for the residuals
instead of a single index. 
For example, for player $i$, there are $n_i$ instances at which we have $\ell_i$ observations. Let
$(e_i)_{k,j}=(e_i)_{(\ell_i+1)(k-1)+j}$ where $k\in\{1,\ldots, n_i\}$ and $j\in
\{1,\ldots,(\ell_i+1)\}$. With the residuals, we can then form estimates
$\hat{B}_{i,k}\in \R^{(\ell_i+1) \times (\ell_i+1)}$ of
$B_{i,k}$ where $\hat{B}_{i,k}$ takes the form 
\begin{equation}\label{eq:Freedmannoise}
    \hat{B}_{i,k}=[(\hat{B}_{i,k})_{l,j})]_{l,j=1}^{\ell_i+1}
\end{equation}
with 
$(\hat{B}_{i,k})_{j,j}=n_i^{-1}\sum_{t=1}^{n_i} e_{t,j}^2$ and
$(\hat{B}_{i,k})_{l, j}=n_i^{-1}\sum_{t=1}^{n_i} e_{t,j}e_{t,l}$ for
$j\neq l$.
We provide this noise structure as an example because in our formulation we
allow for constraints on the players' optimization problems so that for each iteration $k$, we
in fact have multidimensional observations as can be seen in
\eqref{eq:regressormatrix}. 

The second noise structure we consider is adapted from the
$\text{HC}_{4}$
estimator~\cite{cribari2004asymptotic} and is given by 
\begin{equation}\label{eq:HC4noise}
\textstyle\hat{G}= \text{diag}
\left(\frac{e_{1}^{2}}{(1 - b_{1})^{\delta_{1}}},\frac{e_{2}^{2}}{(1 -
b_{2})^{\delta_{2}}}, \cdots, \frac{e_{n_d}^{2}}{(1 -
b_{n_d})^{\delta_{n_d}}}\right)
\end{equation} where 
$\delta_{i} = \text{min} \left\{4,
  n_db_{i}/(\sum_{i=1}^{n_d}b_{i})\right\}$
 and the $b_{i}$'s are the diagonal elements of $B=X(X^{\top}X)^{-1}X^{\top}$. 
With this structure, the penalty for each residual increases with
$b_{i}/\sum_{j=1}^{n_{d}}b_{j}$. 
As with the previous noise structure, 
we use the fitted cOLS estimator $\hat{\beta}_{\text{cOLS}}$
and residuals to get an initial $\hat{G}$. 
We selected to present this noise structure because it is computationally
efficient compared to many other noise structures.

In both cases, we substitute the inferred noise,
$\hat{G}$, into the cGLS statistical model \eqref{eq:gls} to get the one--step
constrained Feasible GLS (cFGLS) 
estimators. We iterate between the estimation of $\hat{G}$ and $\beta_{\cFGLS}$ either until convergence or
for a fixed number of iterations to prevent overfitting. To resolve this trade-off and find the optimal iteration size we adopt a simple cross validation method.

\subsection{Boosting with Ensemble Methods}
In this subsection, we describe several ensemble methods. Combined
    with a bootstrapping process, ensemble methods not only boost the size of
    what can often be a small data set in practice but also 
allow us to improve the estimator performance and explore the bias--variance
tradeoff.

\subsubsection{Bootstrapping and Bagging}

Bootstrapping is a technique for asymptotic approximation of the bias
and standard error of an estimator in a complex and noisy statistical
model~\cite{freedman2009statistical},\cite{Hastie09}. 
We employ \emph{wild bootstrapping} to generate a \emph{pseudo-data set} from
which we generate several weak estimators that we then combine using
 \emph{bagging}.
While we assume that $E(Y | X) = X\beta$, we also allow for
heteroskedasticity by conditioning on the residual transformations that we
imposed in the noise structure.  
Wild bootstrapping is 
a technique of parametric bootstrapping that is consistent with
heteroskedastic inference and cFGLS data generation.
 
The bootstrapping process can be described in two steps: First,
we fit our cFGLS model which gives us $\hat{\beta}_{\text{cFGLS}}$. Then, 
 generate $N$ replicates of \emph{pseudo--data} 
 using the data generation process
 \begin{equation}\label{eq:wildbootstrap}
   \tilde{Y} = X\hat{\beta}_{\cFGLS} + \Phi(e) \vep,
\end{equation}
where $\tilde{Y} \in \R^{n_d}$ is the new observation vector
(pseudo-observations), 
$\hat{\beta}_{\cFGLS} \in \R^{n_d}$ is the cFGLS estimator,
$\vep \sim N(0,I^{n_d \times
n_d})$, $e \in \R^{n_d}$ is the residual vector given by $e = \tilde{Y} -
X\hat{\beta}_{\cFGLS}$  and $\Phi:\R^{n_d}\rar \R^{n_d}$ is a
nonlinear transformation such that $\Phi(e)=\hat{G}^{\frac{1}{2}}\in \R^{n_d \times n_d}$.
Since $E
(\Phi(e) \vep | X) = \Phi(e)E(\vep | X) =
\Phi(e)E(\vep) = \mathbf{0}_{n_d \times n_d}$, using the data generation process in
\eqref{eq:wildbootstrap}, we resample from i.i.d variables. 

Bagging in regression models and trees is a
technique for reducing the overall variance~\cite{Hastie09}.  
Using the $N$ replicates of pseudo--data generated by wild bootstrapping, we train
$N$ different models. We combine the resulting
bootstrapped estimators by averaging:
\begin{equation}\label{eq:baggedlearner}
  \textstyle \hat{\beta}_{\text{bag}} = \frac{1}{N} \sum_{j=1}^{N}
  \hat{\beta}_{\cFGLS, j}
\end{equation}
where 
$\hat{\beta}_{\cFGLS,j}$ is the estimator using the $j$--th pseudo--data sample. Bagging works efficiently with high variance
models and does not hurt the overall performance of the statistical model. We
refer to the bagged estimates as bagged mega-learners since they combine several weak learners/estimators. Using wild bootstrapping, the empirical covariance matrix of
    $\hat{\beta}$ is an
asymptotic approximation of the covariance matrix and 
is given by
 \begin{equation}\label{eq:wildbootstrap}
   \textstyle   \hat{C}_{\beta}=\frac{1}{N}
   \sum_{j=1}^{N} \left( \hat{\beta}_{\text{cFGLS},j} -
   \hat{\beta}_{\text{bag}} \right) \left( \hat{\beta}_{\text{cFGLS},j}-
\hat{\beta}_{\text{bag}}
\right)^{\top}.
\end{equation}
Asymptotic estimation of the empirical covariance matrix reveals hidden
structures between players and is what we leverage in the correlation utility
learning procedures.

\subsubsection{Bootstrapping and Bumping}

In a similar fashion as the bagging ensemble method, we combine bumping---a method for fitting cFGLS estimators by using a random
search over the model space~\cite{tibshirani:1999aa}---with the wild bootstrapping generated pseudo-data.
In particular, we apply a \textit{stochastic search} over several different
statistical models coming from a similar data process---i.e.~the data process in \eqref{eq:wildbootstrap}.


 We add the original training data sample to the $N$ replicates of pseudo-data
 generated by the wild bootstrapping
 process and we use this data to estimate $N+1$ cFGLS estimators.
 We evaluate these estimators on the training set and select the one with the
 least training error. The cFGLS bumping estimator is given by
\begin{equation}\label{eq:bumpinglearner}
  \hat{\beta}_{\text{bump}} = \argmin\limits_{\hat{\beta}_{\cFGLS,j}} \| \tilde{Y} - X
  \hat{\beta}_{\cFGLS,j} \|^{2}_{2}
\end{equation}
where $\hat{\beta}_{\text{cFGLS},j}$'s are the cFGLS estimators from
derived from the bootstrapped data. 

\subsubsection{Gradient Boosting}
We combine $L_2$--gradient boosting---which is a repeated least squares fitting
of residuals \cite{friedman2001greedy}---with cFGLS. Gradient boosting is a boosting technique that uses an $L_{2}$ loss function
combined with a gradient descent update method for combining weak learners
at each iteration. 
Boosting estimators are trained in sequence using a
weighted version of the original data set. 
In general, boosting methods are extremely useful for combining models by
incrementally training each new model by emphasizing the errors of the previous
training instances. They are
used extensively in classification methods such as logistic
regression and support vector machines. 

Repeated residual fitting is applied until we reach iteration $m_{\text{stop}}$,
a stopping criteria selected using Akaike
Information Criterion (AIC) to avoid overfitting~\cite{hurvich1998smoothing}.
.
The procedure is detailed in Algorithm~$1$.

\begin{algorithm}[h]
  \label{alg:CFGLSadaboost}
\centering
\caption{$L_2$--gradient boosting with cFGLS} 
\begin{algorithmic}[1]
\Function
{cFGLSgradboost}{$X$,$Y$,$\nu$}
\State $\hat{H} \gets X(X^{\top}X)^{-1}X^{\top}$\Comment{compute $\hat{H}$ matrix}
\State $\nu \gets s \in (0,1]$\Comment{set shrinkage (updating) parameter}
\State $m_{\text{stop}} \gets 1$\Comment{iteration number}
\State choose $M_{\max}$ \Comment{upper iterations bound}
\State $\AIC_{\text{list}}\gets$[ ]\Comment{create empty list}
\State \underline{\textbf{Compute stopping iteration time $m_{\text{stop}}$:}} 
\While {$m_{\text{stop}} < M_{\max} $}
\State $R_{\ms} \gets (I_{n_d\times n_d} - \nu \hat{H})^{\ms}$
\State $B_{\ms} \gets (I_{n_d\times n_d} - R_{\ms})$
\State $\sigma_{\ms}^{2} \gets n_d^{-1}\sum_{i=1}^{n_d}(Y_{i} - (B_{\ms}Y)_{i})^{2}$
\State $AIC_{\ms} \gets \left(\log{\sigma_{\ms}^{2}} + \frac{1 +
(\Tr(B_{\ms}))/n_d}{1 - (\Tr(B_{\ms}) + 2)/n_d}\right)$
\State $AIC_{\text{list}}$.append$(AIC_{\ms})$
\State $\ms \gets \ms+1$
\EndWhile
\State $\hat{M} \gets \arg\min AIC_{\text{total}}$\Comment{find minimum point}
\State $\hat{\beta}_{\cFGLS}\gets $ estimate of $\beta_{\text{cFGLS}}$ using
cFLGS
\State $e_{\FGLS} \gets Y - X\hat{\beta}_{\cFGLS}$\Comment{residuals estimation}
\State $e \gets e_{\cFGLS}$\Comment{initialize residuals}
\State $k \gets 1$\Comment{iteration index}
\State $\beta_{\text{boost}} \gets \hat{\beta}_{\cFGLS}$\Comment{initialize cFGLS
boosted learner}
\State \underline{\textbf{Compute boosted learner $\beta_{\text{boost}}$:}} 
\While {$k < \hat{M}  $} 
\State $\beta_{i} \gets (X^{\top}X)^{-1}X^{\top}e$\Comment{residuals fitting}
\State $\hat{\beta}_{\text{boost}} \gets \hat{\beta}_{\text{boost}} + \nu \beta_{i}$\Comment{update formula}
\State $e \gets Y - X\hat{\beta}_{\text{boost}}$\Comment{residuals update}
\State $k \gets k+1$ 
\EndWhile
\EndFunction
\end{algorithmic}
\end{algorithm}


\section{Application to Bertrand-Nash Competition}
\label{sec:bertnash}
Let us illustrate the framework and its performance of the robust utility
learning framework before moving on by applying it to estimate market demand
functions under Bertrand-Nash equilibrium (see,
e.g.,~\cite{prakash:2009aa,berry:1995aa, bertsimas:2015aa}). The toy model can be thought of as an abstraction
of Bertrand-price setting for commodities such as oil, gas, and
coal~\cite{gabriel:2013aa,ledvina:2011aa}.


Consider two firms competing to sell their product by setting the price $p_1$
and $p_2$ for firm $1$ and $2$, respectively. The firms utility functions are
their revenue, i.e.
    $f_i(p_1,p_2)=p_iD_i(p_1,p_2,\xi)$ where $D_i$ is the demand function for
    firm $i$ and $\xi\sim \mc{N}(1.5, 0.5)$ is a random variable that
    captures the fact that demand is dependent on economic indicators in
    addition to the prices set by the firms. In this stylized example, we
    consider linear demand functions given by
    \begin{equation}
        D_i(p_1,p_2,
        \xi)=\theta_{i,1}+\theta_{i,2}p_1+\theta_{i,3}p_2+\nu\xi
        \label{eq:paramdemand}
    \end{equation}
    where $\theta_i=(\theta_{i,j})_{j=1}^3$ are unknown parameters to be
    estimated and $\nu=1.5$ is a known parameter. The prices are
    constrained to be in the interval $[0,\bar{p}]$ where $\bar{p}\in \mb{R}_+$
    is the upper bound. 
We let $\theta_1=(-1.0,0.5,-1)$ and $\theta_2=(0.3,-1,0.3)$ be the ground truth
values for the parameters we wish to estimate. Thus,
$\bar{f}_i(p_1,p_2)=\nu\xi$ and examining the marginal revenue functions
$D_if_i(p_1,p_2)$ we have that
$\phi_1(p_1,p_2)=[1\ 2p_1\ p_2]^\top$, and $\phi_2=[1\ p_1\ 2p_2]^\top$. 

In order to generate the data set we add a noise term $\vep\sim\mc{N}(0,0.5)$ to
the marginal revenue functions, i.e.~$D_if_i(p_1,p_2)+\vep$, and solve for the
Bertrand-Nash equilibrium. We simulate the game between the firms $600$ times. In the
robust utility learning framework, for this example, we employ the
HC\textsubscript{4} noise structure and compute the cOLS, cFGLS, bagging,
boosting and bumping estimators. We use a $10$--fold cross validation proceedure 
to prevent over-fitting. Table~\ref{tab:bertnash} contains error using two
metrics for both firms. Figure~\ref{fig:bertnash} shows the forecast for part of
the testing set using cOLS
and each of
the ensemble methods as compared to the ground truth. While bagging
performed best for firm $1$ and boosting for firm 2 in the particular
instantiation of this
toy example, the performance more generally is dependent on the noise structure in the demand and
marginal revenue functions, the sample size, and the dynamics between the two
firms.  
However, it is interesting to point out that as we
increase the variance on $\xi$, each of the ensemble methods performance stay
relatively the same yet the cOLS error increases significantly.


{ \begin{table}[]
\centering
\caption{\scriptsize Mean Square Error (MSE) 
 of forecasting
using the proposed robust utility learning methods vs cOLS estimators for
Bertrand-Nash competition. The best performing method is
indicated in bold text for each of the firms.}
 \label{tab:bertnash}
\begin{tabular}{|l|c|c|c|c|}
\hline
{\small \emph{\textbf{Firm 1}}} & \textbf{bagging} & boosting & bumping & cOLS  \\ \hline
\textit{ MSE }&\textbf{ 0.05 }& 0.51 & 0.65 & 1.62  \\ \hline \hline
{\small \emph{\textbf{Firm 2}}} & bagging & \textbf{boosting} & bumping & cOLS  \\ \hline
 \textit{ MSE }& 1.58 & \textbf{0.71} & 0.89  & 2.54 \\ \hline
\end{tabular}
\end{table}

\begin{figure}[h]
    \centering
    \includegraphics[width=\columnwidth]{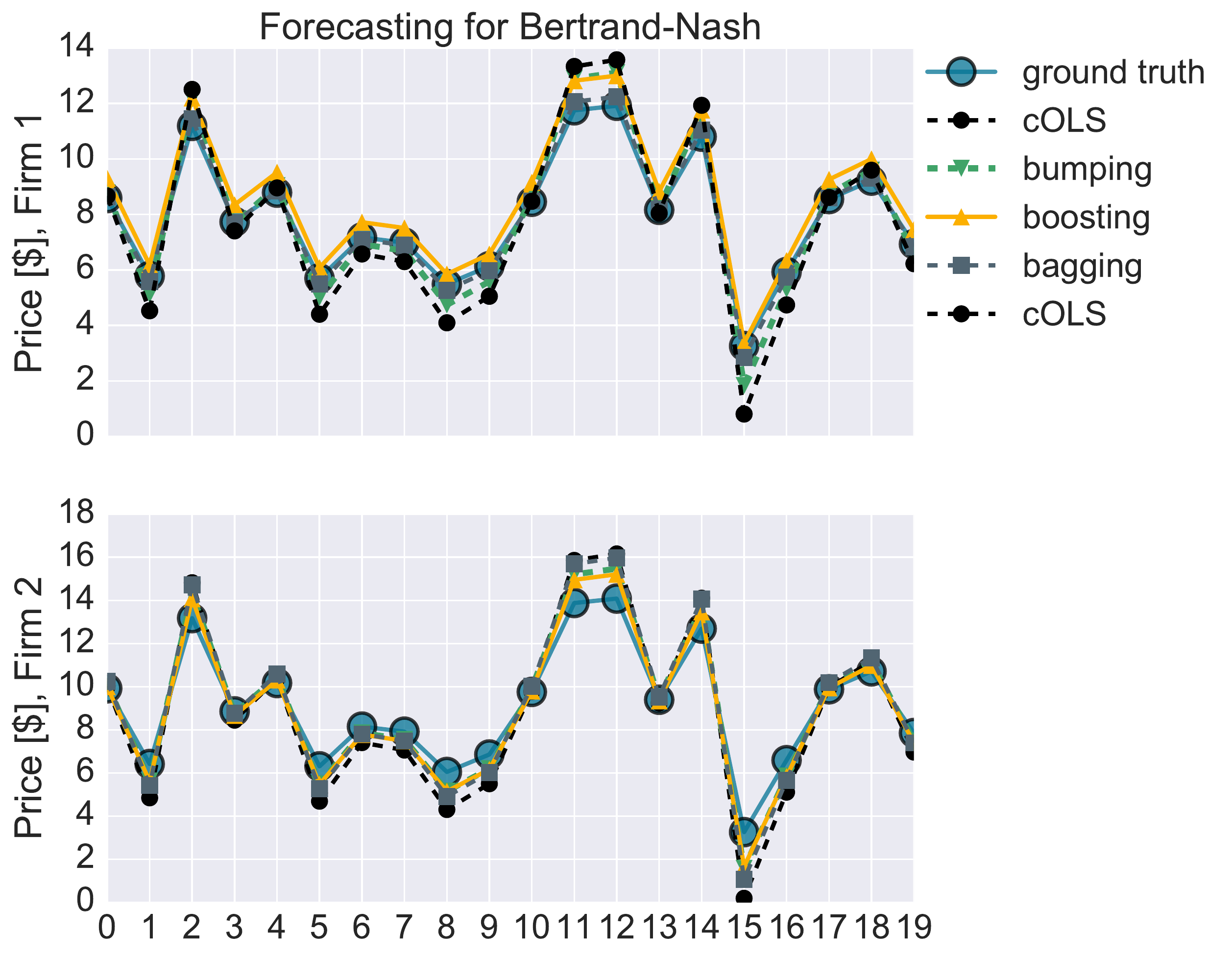}
    \caption{Forcast for Firms 1 \& 2 using cOLS and each of the
    ensemble methods. The ground truth prices 
  are depicted by the \textbf{\textcolor{blue!50!green!85!black}{blue dots}}; the
  cOLS forecasts are depicted in \textbf{black}, the bagging forecasts are
  depicted in \textbf{\textcolor{gray!85!blue!85!black}{gray}}, the 
bumping forecasts are depicted in
  \textbf{\textcolor{green!70!black}{green}}, and the
 boosting forecasts are depicted in
  \textbf{\textcolor{yellow!80!orange!90!black}{gold}}.}
    \label{fig:bertnash}
\end{figure}

\section{Correlated Utility Learning}
\label{sec:corutility}
In this section, we describe how learned correlations between players can be
leveraged to \emph{boost} estimator performance. We add a second step to the
estimation procedure in which we craft a new game where {players'} utilities 
are composed of their original estimated utility plus
some combination of other players' utilities weighted by the estimated
correlation between players.

When the correlations between \players are positive, we are creating what we refer to as \emph{pseudo-coalitions} since
players are not \emph{explicitly} agreeing to collude in the game but rather are doing
so \emph{implicitly}. The degree of \emph{coalition} is discovered by the robust
utility learning process through estimating the empirical
covariance $\hat{C}_{\beta}$, i.e.~asymptotic approximation of the covariance matrix---of
$\hat{\beta}_{\text{est}}$ where we use the notation
$\hat{\beta}_{\text{est}}$ to abstractly denote the estimator derived from whichever
of the methods described in the previous section is employed. 
On the other hand, when the correlations between \players are negative, by
combining their utilities we aim to take advantage of active players' richer data
sets in predicting the behavior of \players with less variation and frequency in
their observed actions.

 We
refer to the learned utility---$\hat{f}_{i}$ for player $i$---from the robust utility learning framework as the
\emph{base utility} and it  
is given by
\begin{equation}
    \hat{f}_{i}(x_i, x_{-i};\hat{\theta}_i)=\bar{f}_i(x_i,x_{-i})+\la\phi_i(x_i, x_{-i}), \hat{\theta}_{i}\ra
  \label{eq:nominal}
\end{equation}
where $\hat{\theta}_i$ is extracted
from 
$\hat{\beta}_{\text{est},i}$.

Using the correlations we learn when we estimate $\hat{f}_{i}$, we
construct a new utility $\hat{g}_{i}$ by combining scaled versions of a subset
(potentially all) of
the other
agents' utilities that are correlated with agent $i$. We formulate an
optimization problem to deterimine the scaling coefficients. 
The correlated utility $\hat{g}_i$ for player $i$ is given by 
\begin{align}
\textstyle    \hat{g}_{i}(x_i,x_{-i}) &=\textstyle \sum_{j\in \mc{K}_i}
z_{i,j}\sigma_{i,j}\bar{f}_i(x_i,x_{-i})\notag\\
    &\quad+\sigma_{i,j}\la\phi_i(x_i,x_{-i}),\hat{\theta}_j\ra
    \label{eq:corest-1}
\end{align}
where $\mc{K}_i\subset\mc{I}_i$ a subset of the players correlated with player
$i$, $\sigma_{i,i}$ is the estimated variance of player $i$  determined by the empirical
covariance matrix, $\sigma_{i,j}$ is the covariance between the parameter
estimates for player $i$ and $j$
also determined by the empirical covariance matrix, and $z_{i,j}$ are scaling
constants to be optimized.
We refer to the resulting game as an \emph{approximated correlation game}\footnote{We remark that there exists an equilibrium concept called \emph{correlated
equilibrium}~\cite{aumann:1974aa} which generalizes a Nash
equilibrium by characterizing correlations between randomized strategies; we
mention this only to alleviate  any potential confusion. 
The equilibrium concept we utilize for the
approximated correlation game is still a pure Nash equilibrium and there is no
coordinating mechanism.}.

{Given the form of $\hat{g}_i$, our goal is to select the scaling
constants $z_{i,j}$ in order to
reduce the forecasting error.} Analogous to the base utility learning framework
presented in Section~\ref{sec:baseutility}, using our training data, we
formulate a convex optimization problem using optimality
conditions on each player's individual optimization problem where we assume that
player $i$ is optimizing $\hat{g}_i$ with respect to its own choice variable
$x_i$.   
In particular, we
solve a convex optimization problem formulated as follows. Define the vector $z_i\in \mb{R}^{|\mc{K}_i|}$ by
  $z_i=( {z_{i,j}})_{j\in \mc{K}_i}$
  and let $z=(z_i)_{i\in \mc{I}}$.
For player $i$'s
optimization problem $\max\{\hat{g}_i(x_i,x_{-i})|\ x_i\in
\mc{C}_i\}$, let the residual of the stationarity
condition be given by 
 \begin{align}
  \label{eq:dstatzero-1}
  \textstyle  r_{\text{s},i}^{(k)}(z_i, \mu_i;\hat{\theta}) &= \textstyle D_i\hat{g}_i(x_i^{(k)},
  x_{-i}^{(k)})+\sum_{j=1}^{\ell_i}\mu_i^{j}D_ih_{i,j}(x_i^{(k)})
\end{align}
and the residual of the complementary conditions be given by
\begin{align}
  r_{\text{c},i}^{j,(k)}(\mu_i) &=  \mu_i^{j}h_{i,j}(x_i^{(k)}), \ j\in
  \{1,\ldots, \ell_i\}.
\label{eq:dslaczero-1}
\end{align} 
As before, let
$r_{\text{c},i}^{(k)}(\mu_i)=[r_{\text{c},i}^{1,(k)}(\mu_i)\ \cdots\ 
r_{\text{c},i}^{\ell_i,(k)}(\mu_i)]$.
Define $Q_i\in \mb{R}^{n_i\times|\mc{K}_i|}$ 
by
\begin{align}
    \textstyle  Q_i=\left[ \sigma_{i,j} D_{i,i}^2\bar{f}_i(x^{(k)}) \right]_{k=1,j\in\mc{K}_i}^{n_i}.
  \label{eq:meth1Q}
\end{align}
and $q_i\in \mb{R}^{n_i}$ by
\begin{equation}
 \textstyle q_i=\left[ \sum_{j\in \mc{K}_i}\sigma_{i,j}\la
      D_{i,i}^2\phi_i(x^{(k)}), \hat{\theta}_j\ra
  \right]_{k=1}^{n_i}.
  \label{eq:vector1}
\end{equation}
Then,  we have the following convex optimization problem {to determine
    the scaling factors $z_{i,j}$}:
\begin{equation}
 \begin{aligned}
   &   \min\limits_{z,\mu}\textstyle\sum_{i=1}^p
   \textstyle\sum_{k=1}^{n_i}\chi_i(r_{\text{s},i}^{(k)}(z_i, \mu_i; \hat{\theta}),
  r_{\text{c},i}^{(k)}(\mu_i))\\
  &  \text{s.t.}\ \ Q_iz_i+q_i\leq 0,  \ \mu_i\geq 0\  \ \forall \ i\in \mc{I}
\end{aligned}
\tag{P'}\label{opt-Pprime}
\end{equation}
Solving \ref{opt-Pprime} gives us estimated correlated utilities $\hat{g}_i$ for
each $i\in \mc{I}$ that we then use to forecast the players' decisions.

\section{Application to Smart Building Social Game}
\label{sec:apptosoc}
We now specialize the robust and correlated utility learning frameworks to the
smart building social game
.
\begin{figure*}[ht]
\center   
\subfloat[\label{fig:points}]{\includegraphics[width=0.6\textwidth]{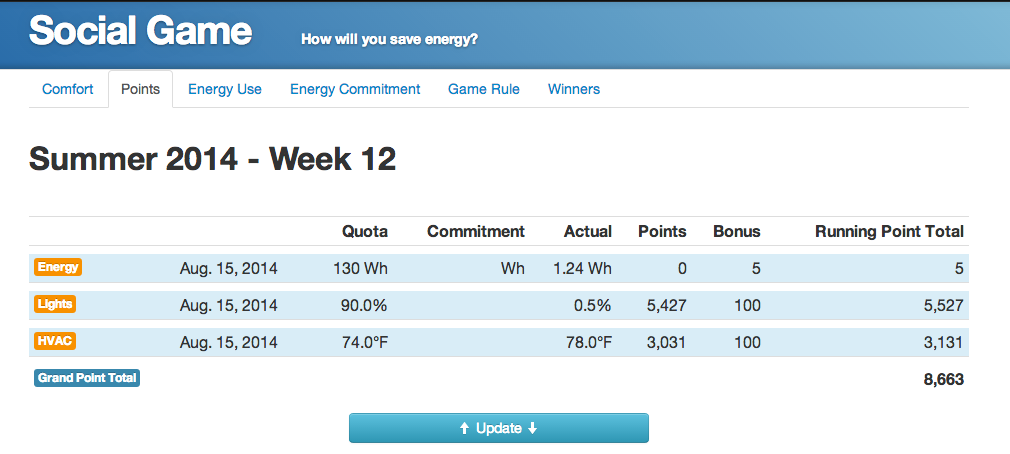}}
\subfloat[\label{fig:votes}]{\includegraphics[width=0.15\textwidth]{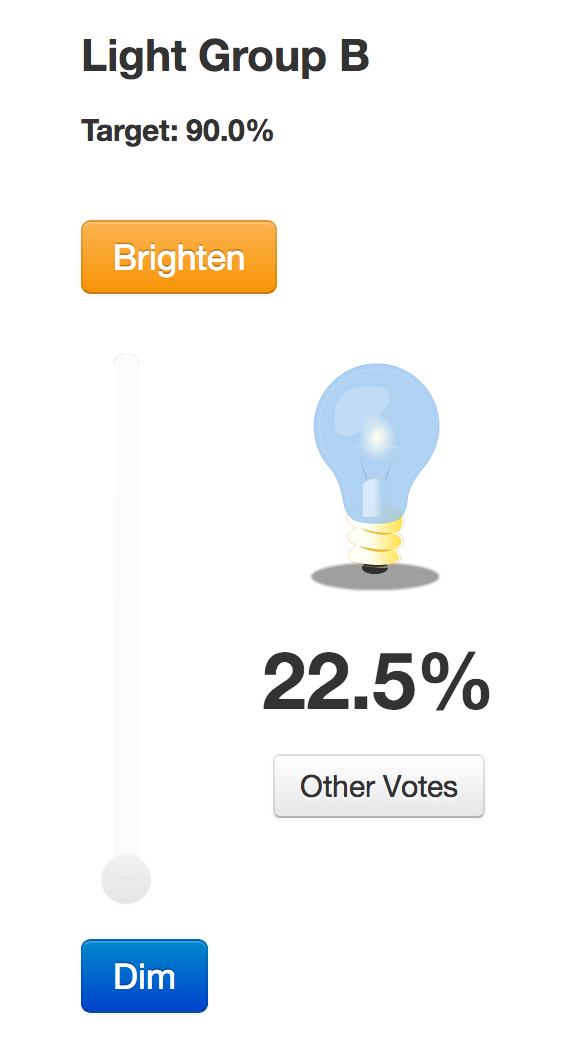}}

  \caption{Graphical user interface (GUI) for energy based social game: (a)
  Display, in table form, of points and votes for energy consumption, HVAC, and lights. (b)
  Display of the GUI for
  logging lighting setting preferences. }
  \label{fig:all}
\end{figure*}

\begin{figure}[ht]
\center\subfloat[\label{fig:zones1}]{\includegraphics[width=0.3\textwidth]{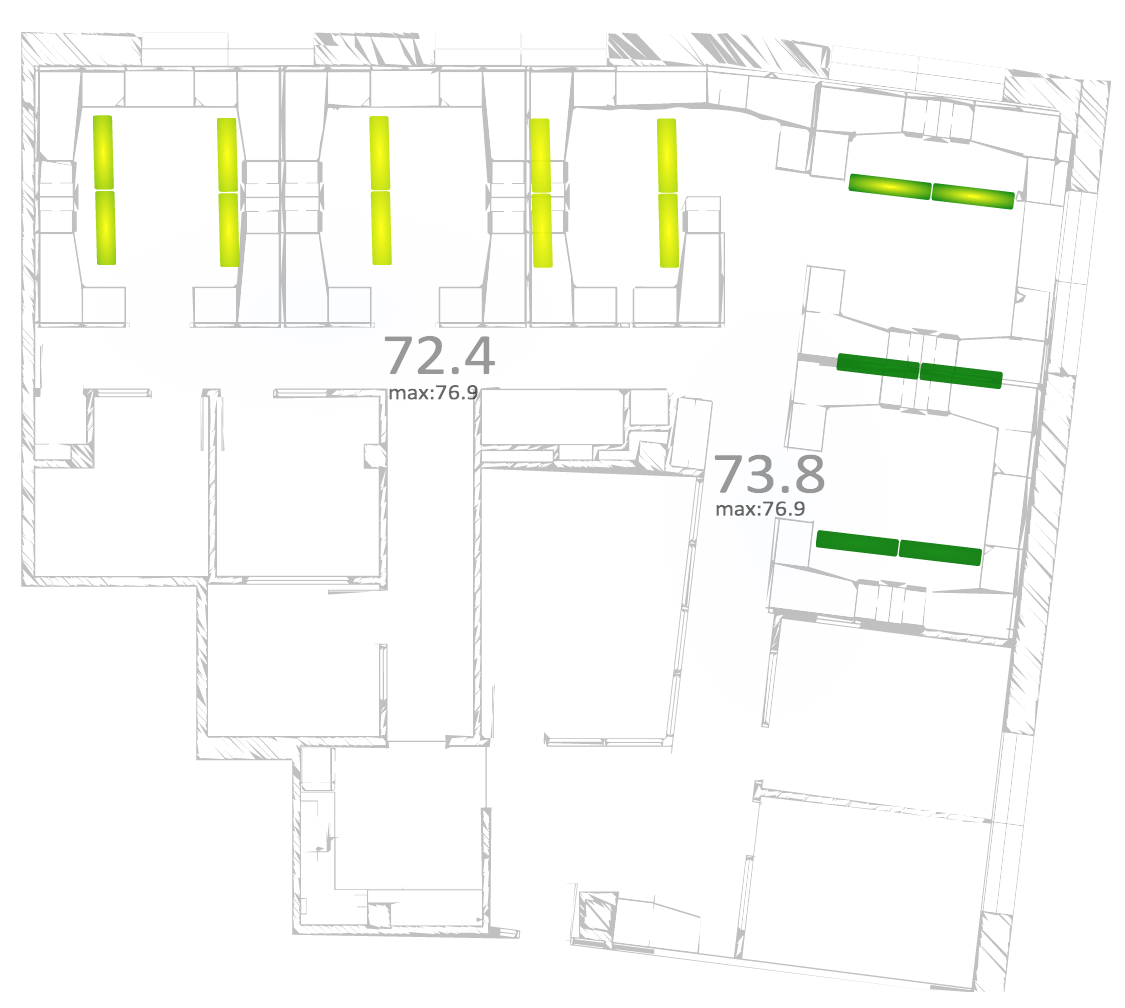}}
      
\subfloat[\label{fig:zones2}]{\includegraphics[width=0.3\textwidth]{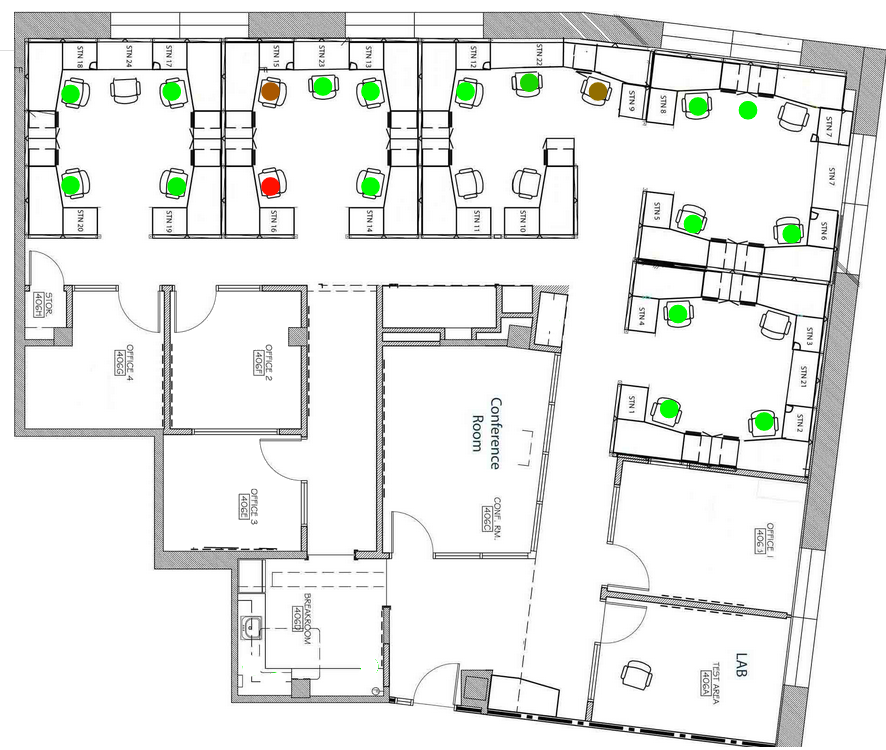}}
         \caption{
Occupants can access a variety of information when they log into the social
  game portal, including various displays of energy consumption by other
  participants in the game:  (a) Display
  of current light level and temperature in the collaboratory space; energy
  efficiency of the lights is coded by color where light green
  indicates \emph{higher energy efficiency}. (b) Display of collaboratory floor
  plan with dots indicating where present and participating players sit. Players
  not in the office are excluded from the game. The color of the
  dot indicates the level of energy efficiency of the player as compare to the
  other participants; green indicates higher efficiency while red indicates
  lower efficiency.}
  \label{fig:all-1}
\end{figure}

\subsection{Social Game Experimental Set-Up}
\label{sec:experiment}
Our experimental setup is in a collaboratory space---an open, shared work space with
cubicles---within the CREST center on the UC Berkeley campus. We crafted a social game
such that occupants in this collaboratory
freely vote according to their usage preferences of shared resources and are
rewarded with points based on how {energy efficient} their strategy is in
comparison with the other occupants.   We employ a lottery mechanism consisting
of three Amazon gift cards executed bi-weekly to reward occupants; occupants with more points are
more likely to win the lottery. 

The office is divided into five lighting zones and 
two heating, ventilating, and air conditioning (HVAC) zones. In this space,
there is a total $20$ occupants who are eligible to participate in
the social game. 
If the occupants are not present in the office, they are
excluded from the game at that time instant. When they arrive at the office, they can
rejoin the game. To enforce the rule that those who are not
present in the space cannot vote remotely, we executed a simple presence
detection algorithm based on their power usage~\cite{jin2014presencesense,jin2017virtual}.

We have installed a
Lutron\footnote{\tt http://www.lutron.com/en-US/Pages/default.aspx} system for
precise control of the lighting setting (dim level of the lights) in the office 
as well as  desk--level energy monitoring devices
(i.e.~ACME wireless sensors~\cite{jiang2009design}) to meter the energy usage of
each occupant. In addition, we have modified the HVAC system so that it can be
precisely controlled. We have verified prior to our experiment that implemented
control of these systems results in expected performance. 

We have developed a platform to interface with the occupants as well as manage
and process collected data. The platform includes a  web portal and mobile
app that the occupants may use to participate in the game. 
It also allows for occupants to visualize different aspects of
the social game---e.g., the lighting setting and the energy
efficiency level of different occupants or the entire building---as well as view the point level and
historical voting record
of other occupants among many other statistics. 
Figure~\ref{fig:all} shows the user interface for
viewing points and logging votes. Figure~\ref{fig:zones1} shows a visualization
of the current light level using a green--to--red color scale with green being
more energy efficient. The current temperature is also displayed. Figure \ref{fig:zones2}
shows a visualization of each present and participating occupant's energy
efficiency level.

In this paper, we report on a social game experiment conducted based only the
lighting shared resource\footnote{We remark that while our experimental platform
    is capable of conducting a social game that includes lights, HVAC, and personal
energy consumption, we only report on an experiment that focuses on lighting
in order to isolate combined effects from these different resources. In
on-going experiments, we are examining all aspects jointly.}.
Prior to the start of the social game experiment, the lighting setting was $90$\% of the
maximum possible lighting setting.
At the start of the social game experiment, we set a default lighting setting
which acts as the suggested lighting setting and is the dim level setting in the
office if, e.g., no occupants are participating
in the game.
Throughout the game, we adjust the default lighting setting as well as the
points. The lottery mechanism coupled with the points we
distribute compose the \emph{incentive component} of the feedback to the
participants while the default
lighting level is the \emph{physical control component} of the feedback.
 These two mechanisms act as our control inputs and our feedback
mechanism to the participants. We seek to design them by taking into consideration the
preferences of the participants. In this way, these mechanisms close the loop around the
participant and with our proposed utility learning scheme, these mechanisms can be
modified to encourage more energy efficient resource consumption.

The game is designed to leverage interactions amongst
occupants, who win points based on how energy efficient their lighting vote is compared
to others. 
  An occupant's vote is for the lighting setting in their zone as well as for
neighboring zones.   The occupants select their desired lighting setting in the
continuous interval $[0,100]$ where each value represents the percentage of the
maximum lighting setting possible in the space. 
The occupants can vote as frequently as they like and the average of all the
occupants' current votes sets the implemented lighting setting  in the collaboratory.
An occupant can leave the
lighting setting as the default level after logging in or they can change it 
 depending on their preferences and other 
environmental factors that may affect their choice. 

The experimental trials reported on in this paper were conducted over
    the period of $285$ days\footnote{The period of the experiment was $2014/3/3$--$2014/12/14$.}.
Experiments with $4$ different default levels,
$\{10\%,20\%,60\%,90\%\}$, were conducted, covering a spectrum of lighting
conditions. 
Since occupants were allowed to vote whenever they chose, their response rate
per day varies. The data set we collected consists of occupant
votes (meaning the lighting level they select) over the period of investigation
as well as the points that were distributed to each occupant.
We collected 6,885 votes over the period of the experiment.

\subsection{Brief Background}
\label{sec:litreview}
In order to place the work pertaining to building energy efficiency in the context of the state
of the art, we briefly overview existing approaches. 


Recognizing that HVAC systems are responsible for a large portion of building
energy consumption, many control theoretic approaches such as \cite{Aswani:2012kx,jin2015sensing} derive model predictive and distributed control polices for HVAC
systems. 
While these control
theoretic approaches make efforts to account for the presence of occupants, they
tend to ignore occupant
behaviors and, more importantly, their heterogeneous preferences.

There are other works that make strides towards incorporating behavioral models
of occupants; e.g., the authors of \cite{boman:1998aa} employ a
multi-agent systems approach to develop a framework for incorporating occupant
comfort
preferences and the authors of \cite{bourgeois:2006aa} develop behavioral models
for lighting usage.
In a more active approach, the authors of~\cite{song:2013aa} 
develop a \emph{collaborative setting
definition paradigm} in which occupants and facilities managers submit
preferences and requirements and a rule engine tries to resolve them in order to
create a universal control policy.
While occupants' preferences are taken as inputs to the building control design,
it is not clear that it is possible to satisfy all the occupants' comfort
preferences simultaneously with those of the facilities manager; hence, the
misalignment between preferences and incentives remains.

%

In our approach, on the other hand, we leverage a social game that creates
(friendly) competition between users and employs incentives to resolve conflicting
preferences by compensating users. Within the energy application domain,
gamification has been largely used
for 
education
or awareness~(see, 
e.g.,~\cite{bang:2006aa,simon:2012aa}).
There are works that are closely related to ours in the sense that they also
recognize that occupants are self-interested participants in smart buildings
and try to account for their strategic behavior. For example, in~\cite{li:2014aa}, the authors
develop an interesting scheme for engaging occupants directly in DR.
Analogous to our approach, occupants are modeled as utility maximizers in a game
theoretic context where they 
are
incentivized to curtail their consumption in response to an event. Our approach
differs in that we focus on shared resources such as lighting and HVAC instead
of  personal devices (e.g., desk appliances).
Furthermore, it is assumed in~\cite{li:2014aa} that the type space (i.e.~their preferences) of the users
is a known finite set of two possible values.
We do not assume the facility manager knows the utility function or
the type  of
the users and we propose an algorithm for learning this utility function from
observations of decisions.


While incorporating occupant preferences into building automation is
not novel in and of itself, we propose an innovative algorithm for learning occupant
preferences in competitive environments and, moreover, learn how their actions
are correlated. Such correlations can be leveraged in improving incentive
mechanisms to shape users' preferences thereby providing more flexibility.
 Our method is applied to real-world data from
experimental trials we conducted as opposed to simulations as is the case with
many existing works. 
Furthermore, it is agnostic to the application and could
be applied in general to other scenarios 
in which users are competing for constrained but shared resources. 
For example, the utility
learning method can be easily adapted to learning preferences of individual
buildings interacting with an aggregator or learning preferences of 
drivers seeking on-street parking \cite{jin2017mod}. In each of these cases, there exists a
planner---the aggregator or department of transportation---tasked with managing
a  resource being consumed by self-interested users.


\subsection{Occupant Decision-Making Model}
\label{sec:rusocialgame}

Each agent's vote $x_i$ is constrained to be in the interval $[0,100]\subset
\mb{R}$. Let $\bar{x}$ denote the average of the lighting votes and the setting
that is implement---e.g., at observation instance indexed by $k$,
$\bar{x}^{(k)}=\frac{1}{|\mc{S}^{k}|}\sum_{j\in \mc{S}^k} x_j^{(k)}$.  
We model each agent's utility as being composed of two basis functions that
capture the tradeoff between desired lighting (satisfaction) and desire to
win. The lighting satisfaction an occupant feels may be a function of several
factors including their productivity (ability to perform their job) as well as
physical comfort.
We abstractly
model their desired lighting level using a Taguchi loss function,
  $\psi_i(x_i,x_{-i})=-\left( \bar{x}-x_i \right)^2$,
which is interpreted as
modeling occupant dissatisfaction in such a way that it is increasing as variation increases from
their reported desired lighting setting (their vote)~\cite{taguchi:1989aa}. 

We acknowledge that an agent may have some internal desired lighting level
that is different than its vote; e.g., the agent may realize that 
voting an extreme value pushes the average toward a more
desirable setting. This type of \emph{gaming} results in \emph{moral hazard}
type issues which can be addressed in the incentive design
step~\cite{bolton:2005aa,laffont:2002aa}. Thus,
we set this type of {gaming} aside for the time being, and focus instead on
the unknown preferences---a different kind of asymmetric information that leads
to \emph{adverse selection}---between \emph{lighting} and \emph{winning}. 

Points are distributed by the planner using the relationship
$\rho(x_b-x_i)(\np(x_b-\bar{x}))^{-1}$
where $x_b$ is
the baseline setting for the lights. For the experiment $x_b=90$\%, i.e.~the lighting setting used 
before the implementation of the social game.
However, we model each occupant as having a \emph{winning} basis function
given by
 $\phi_i(x_i,x_{-i})=-\rho c \left( {x_i} \right)^2$
where $\rho$ is the total number of points distributed by the planner and
$c$ 
is a scaling factor that is used primarily to scale the two terms of the
utility function given that we artificially inflate the points offered in order
to increase their appeal to 
 players and thus induce greater
participation\footnote{Inflating the points is a process of \emph{framing}~\cite{tversky:1981aa}---that is, dependent on how
the reward system is presented to agents greatly impacts their participation.
Framing is routinely used in rewards programs for credit cards among many other
point-based programs. The scaling factor $c$ in the winning function removes
the framing effect from the estimation procedure. It is selected to ensure the
scale of the two basis functions are similar.}. 
The form of the winning function can be
interpreted as capturing the perception that by voting zero, the occupant is
selecting the action that will provide the
greatest return of points given that points are awarded based on how energy 
efficient their vote is compared to
others\footnote{We
explored other forms of the winning function including the $\log$ function, a 
quasi-concave function that is typically used to represent how individuals value
money since it represents the diminishing returns property well~\cite{ratliff:2014ac}.
However, the quadratic form of the function we report on here
significantly outperformed other choices so that,
for the purpose of a {prescriptive} model, it captures the agents' perceptions about the point distribution
mechanism and their value more accurately.}.

Hence, the utility functions for the social game are modeled as
$f_i(x_i,x_{-i};\theta_i)=\theta_i\phi_i(x_i,x_{-i})+\psi_i(x_i,x_{-i})$.
The constraint sets $\mc{C}_i$ for each player are determined by the box
constraints on the lighting vote for that player, i.e.
 $\mc{C}_i=\{x_i\in \mb{R}|\ h_{i,j}(x_i)\geq 0, \ j\in\{1,2\}\}$
 where $h_{i,1}(x_i)=100-x_i$ and $h_{i,2}(x_i)=x_i$.

In order to formulate \eqref{opt-P} for the social game application, we need to
determine the admissible parameter sets $\Theta_i$, $i\in\mc{I}$ in such a way 
that we ensure the estimated
utility functions are concave and such that equilibria of the estimated game are
isolated.  We derive a lower bound $\theta_{\LB}$ such that all
$\theta_i\in\Theta_i=\{\theta_i\in \mb{R}|\ \theta_i>\theta_{\LB}\}$, $i\in
\mc{I}$ induce games with these characteristics.
%
%
To this end, we utilize the second derivative condition on players'
utility functions; that is, if for each $i\in\mc{I}$,
$D_{i,i}^2f_i(x)<0$, then the game is concave.
Computing $D_{i,i}^2f_i$ and using some algebra, we have that
$\theta_i>-(c\rho)^{-1}(1-\np^{-1})^2$ where the right-hand
side is 
a negative non-increasing function of $\np$. 
Thus, concavity is ensured regardless
of the number of players by setting $\np = 2$, the minimum number of players in a
non-cooperative game. 
Then, given fixed $\rho$ and $0<\zeta<<1$, the lower bound
$\bar{\theta}_{\LB} = -(4c\rho)^{-1}+\zeta$ will guarantee the estimated game is concave.

If $D\omega(x,\mu)$ is invertible, we know that differential Nash equilibria are
isolated~\cite{ratliff:2016aa}. Hence, we can augment the constraint sets $\Theta_i$ to encode this
condition. Given the structure of
the utility functions, $D\omega(x,\mu)$ is simply the game Hessian $H=[H_{i,j}]_{j,i=1}^\np$ with
$H_{i,i}=D_{i,i}^2f_{i}$ and $H_{i,j}=D_{i,j}^2f_i$.
Hence, if $H$ is invertible, then the differential Nash are isolated; this is
guaranteed for $\np \geq 4$ provided
the constraint defined by $\bar{\theta}_{\LB}=-(4c\rho)^{-1}+\zeta$ using
$\zeta=10^{-2}$. Indeed, let $H(\np)$ denote the game Hessian
as a function of the number of players and note that for a particular $\np$,
with some simple algebra, it is easy to write $H(\np)$ as a off-diagonal matrix
constant matrix such that $H_{ii}=d_i+\alpha$ and $H_{i,j}=\alpha$ where
$d_i=-2(1-1/\np)-2c\rho\theta_i$ and $\alpha=2(\np-1)/\np^2$.
It is
straightforward
to verify by determining the eigenvalues of $H$ as $\np$ varies via the method described
in~\cite{gendreau:1986aa} that for $\np\geq 4$, $H$ will be invertible .  
%
For the social game data, at each observation indexed by $k$, the number of
participating players is at least $4$. Thus, to ensure
concavity and isolated equilibria of the estimated social game, we define 
$\Theta_i=\{\theta_i\in \mb{R}|\ \theta_i>\bar{\theta}_{\LB}\}$
with $\bar{\theta}_{\LB}=-(c\rho4)^{-1}+\zeta$ with $\zeta=10^{-2}$. 


\section{Utility Learning Results}
\label{sec:results}
 \begin{figure*}[ht]
   \center \subfloat[\label{fig:dynamic}]{\includegraphics[width=0.515\textwidth]{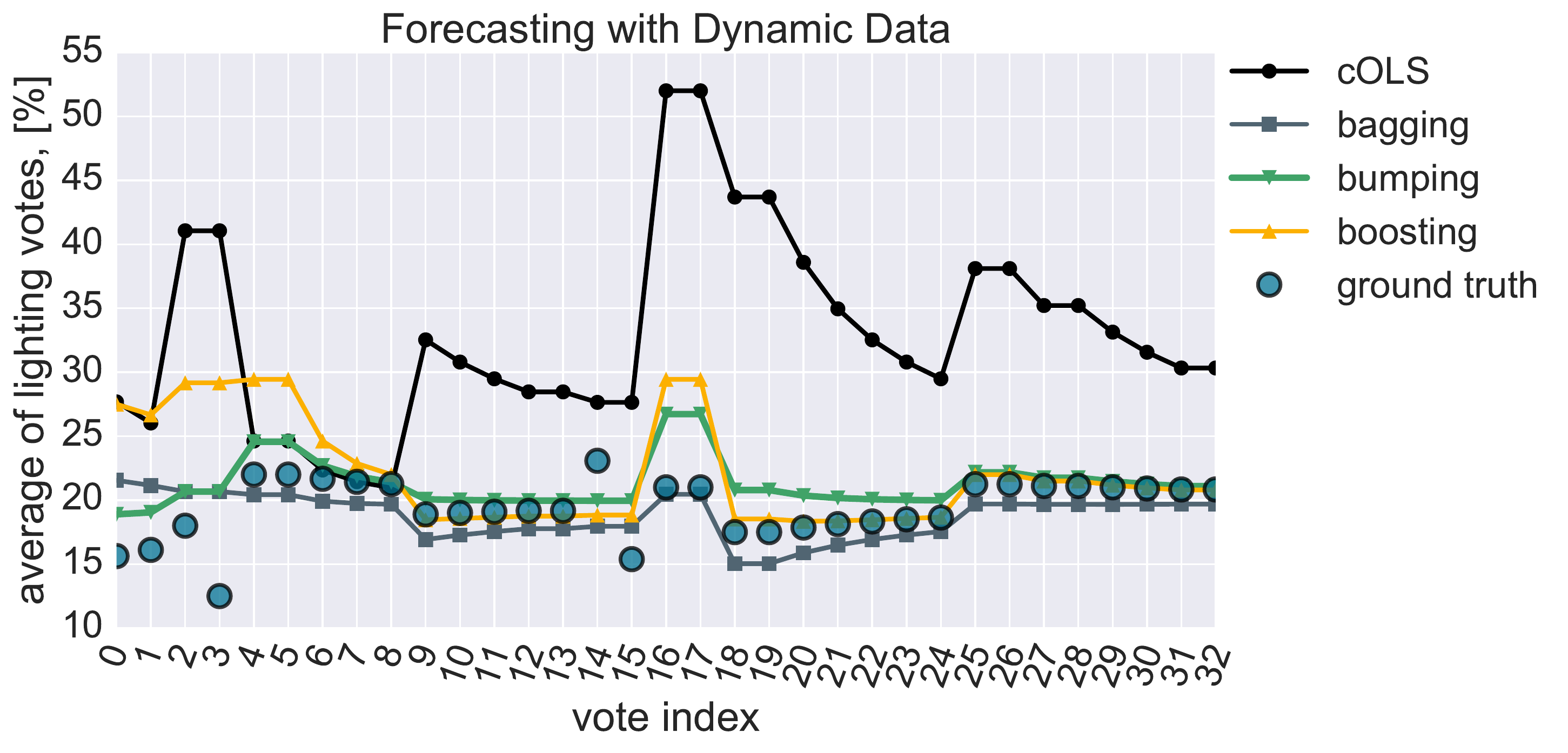}}
    \subfloat[\label{fig:average}]{\includegraphics[width=0.425\textwidth]{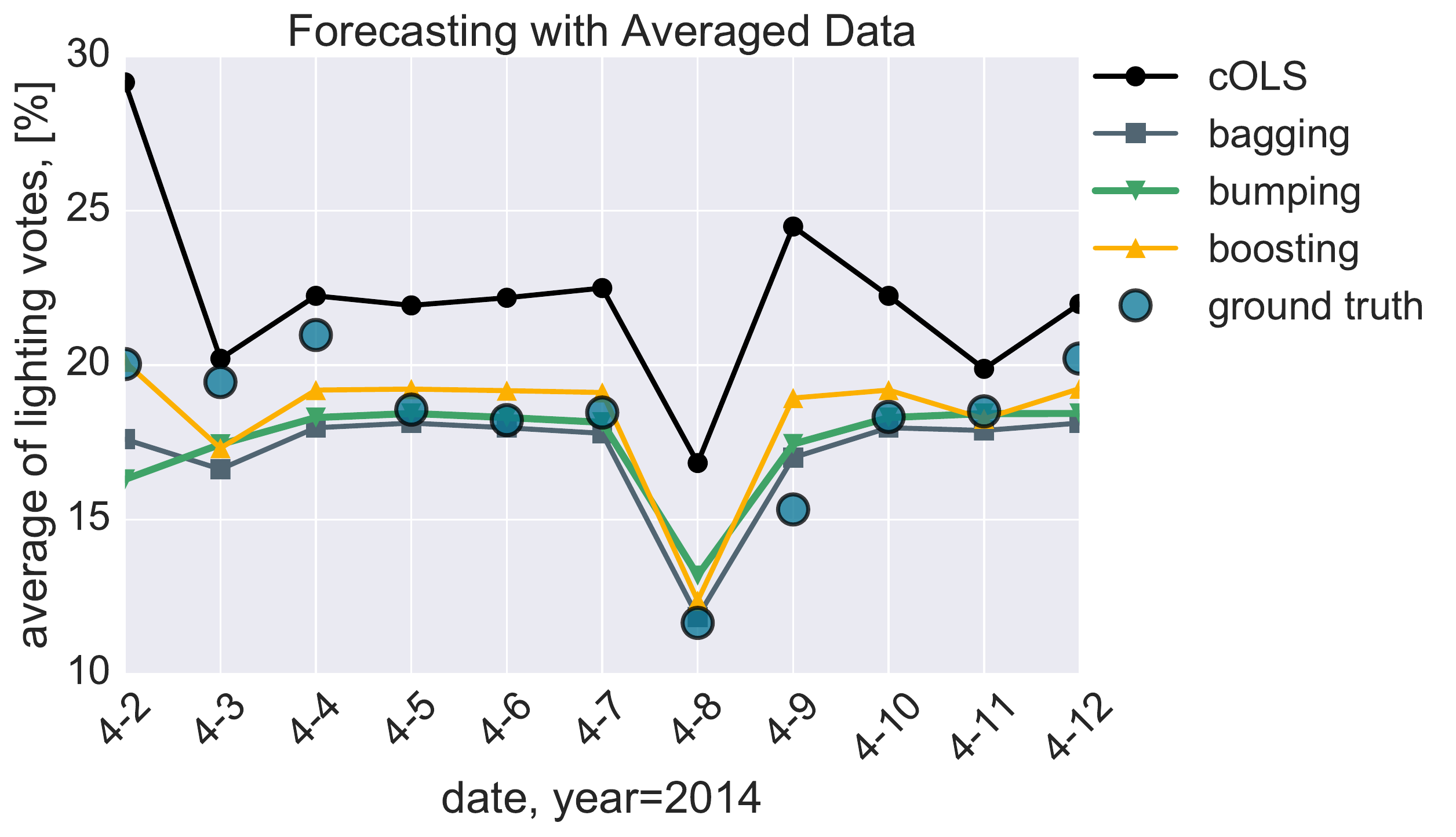}}
  \caption{Forecasting results for (a) dynamic data and (b) averaged
  data for the default lighting setting $20$: For
  the dynamic data, the $x$--axis values indicate the index of when a choice was made by one or more of the occupants (i.e. when
  the implemented lighting setting is changed); the time from one index to the
  next may be several minutes to hours depending on the activity of the
  occupants. For the averaged data, the $x$--axis values are dates (month and
  day).  The ground truth average of the lighting votes
  is depicted by the \textbf{\textcolor{blue!50!green!85!black}{blue dots}}; the
  forecast for cOLS is depicted in \textbf{black}; the
  forecast for bagging is depicted in \textbf{\textcolor{gray!85!blue!85!black}{gray}};
  the forecast for bumping is depicted in
  \textbf{\textcolor{green!70!black}{green}}; the
  forecast for boosting is depicted in
  \textbf{\textcolor{yellow!80!orange!90!black}{gold}}.
The forecast for the robust utility learning methods is approximately near the
ground truth
for both data sets while the cOLS estimates produce Nash equilibria with a large
error. }
  \label{fig:sims}
\end{figure*}

We now present the results of the proposed robust utility learning
method applied to data collected from the social game experiment.

As we previously described, our data set consists of the votes logged by the
players which vote throughout the day. We present estimation results for
the complete data set of all the votes---which we refer to as the \emph{dynamic
data set}---and estimation results for an aggregated data set constructed by
taking the average of a players' votes over the course of each day in the
experiment---this is referred to as the \emph{average data set}.
While this aggregation significantly reduces the size of our data set,
it smooths the players' voting profiles and
increases the size of active players in each game---occupants may arrive or
leave the office when they so choose. This \emph{average data set} also reduces
the computational load, which may be beneficial to a facilities manager in the
incentive design process, especially if the incentive scheme is   quasi-static
and uses historical data to generate the next incentive.
The dynamic data set is much richer, being composed of
 every vote (a total of $6,885$ votes) the occupants made throughout the duration of the
 experiment ($285$ days).
 The time from one vote to the next may be several minutes to hours depending on
 the activity of the occupants. This data set is much larger and thus, increases
 the computational load. However, it allows us to
 extract more distinct player profiles and can support
 real-time incentive design schemes.
  
  We present results for both data sets using data from the period of
  the experiment in which the default lighting setting was $20$\%---the results
  for the other default lighting settings are similar. 
The period of the experiment where the default lighting setting was
$20$\% consisted of $42$ days and thus the size of the averaged data set is
$42$. Over this period there were $220$ votes by occupants, which is the size of
the dynamic data set. 
We divide each of the data sets into training ($80\%$ of the data) and testing
($20\%$ of the data) sets and apply each of the methods discussed in
Section~\ref{sec:robutility}. We apply a $10$--fold cross
validation~\cite{Hastie09} procedure to limit overfitting.


{ \begin{table}[h!]
\centering
\caption{\scriptsize Root Mean Square Error (RMSE), Mean Absolute Error (MAE)
and Mean Absolute Scaled Error (MASE)~\cite{hyndman2006another} of forecasting
using the proposed robust utility learning methods vs cOLS estimators for both
data sets in default lighting setting $20$. The best performing method is
indicated in bold text for each of the data sets, dynamic and average.}
 \label{tab:rmse_20}
\begin{tabular}{|l|c|c|c|c|}
\hline
{\small \emph{\textbf{Dynamic, $\hat{f}_i$}}} & \textbf{bagging} & boosting & bumping & cOLS  \\ \hline
{  RMSE }& \textbf{8.31} & 10.11 & 12.56 & 22.53 \\ \hline
\textit{ MAE }& \textbf{5.20} & 6.55 & 6.38 & 18.35 \\ \hline
\textit{ MASE }&\textbf{ 2.08 }& 6.38 & 2.55 & 7.34\\\hline  \hline
{\small \emph{\textbf{Averaged, $\hat{f}_i$}}} & bagging & \textbf{boosting} & bumping & cOLS  \\ \hline
 \textit{ RMSE }& 2.05 & \textbf{1.68} & 1.96  & 9.36 \\ \hline
 \textit{ MAE }& 1.58 & \textbf{1.31} & 1.48 & 6.01 \\ \hline
 \textit{ MASE }& 0.71 & \textbf{0.59} & 0.67 & 2.69  \\ \hline
\end{tabular}
\end{table}
}

\subsection{Forecasting via Robust Utility Learning} 

We estimate the parameters using cFGLS and the ensemble methods bagging,
bumping, and boosting for both the average and dynamic data sets. For gradient
boosting, we use the HC\textsubscript{4} noise structure
(see~\eqref{eq:HC4noise}) since the \emph{leverage} values $b_{ii}$ of $B$
are larger~\cite{cribari2004asymptotic};
in each of the other methods, we used the block diagonal noise structure
(see~\eqref{eq:Freedmannoise}).

Using the estimated utility functions, we simulate the game using a projected
gradient descent algorithm which is known to converge for concave
games~\cite{flam:1990aa}.
In Figure~\ref{fig:dynamic} and~\ref{fig:average}, we compare the ground truth
voting data to the predictions for each of the learning schemes using the
dynamic and averaged data sets, respectively. Our proposed robust models---i.e.~using the estimated parameters obtained via bagging,
bumping, and boosting---capture most of the variation in the true votes (in both
data sets) and significantly outperform cOLS.
 In Table~\ref{tab:rmse_20}, using three metrics---Root Mean Square Error (RMSE),
Mean Absolute Error (MAE), and Mean Absolute Scaled Error (MASE)---we report
the forecasting error for each of the methods. 

The
estimated models using our robust utility learning methods significantly
reduce the forecasting error as compared to cOLS.
The cOLS method has particularly poor forecasting performance on the dynamic
data set since it does not capture the correlated error terms describing the interactions between users. Moreover, our robust methods 
perform better than cOLS with the averaged data set even though the sample size
is small.

As for the ensemble methods, bagging outperforms the
other three methods when using the dynamic data set.
On the other hand, for the averaged data set, gradient boosting gives the
least forecasting error. This is in large part due to the fact that we use the
HC\textsubscript{4}
noise structure. Since the average data set has been smoothed, we expect less
correlation between players and the HC\textsubscript{4} noise structure captures this.

\ifplot
\begin{figure*}[ht]
  \begin{center}
    \subfloat[\label{fig:contour_2}]{\includegraphics[width=0.41\textwidth]{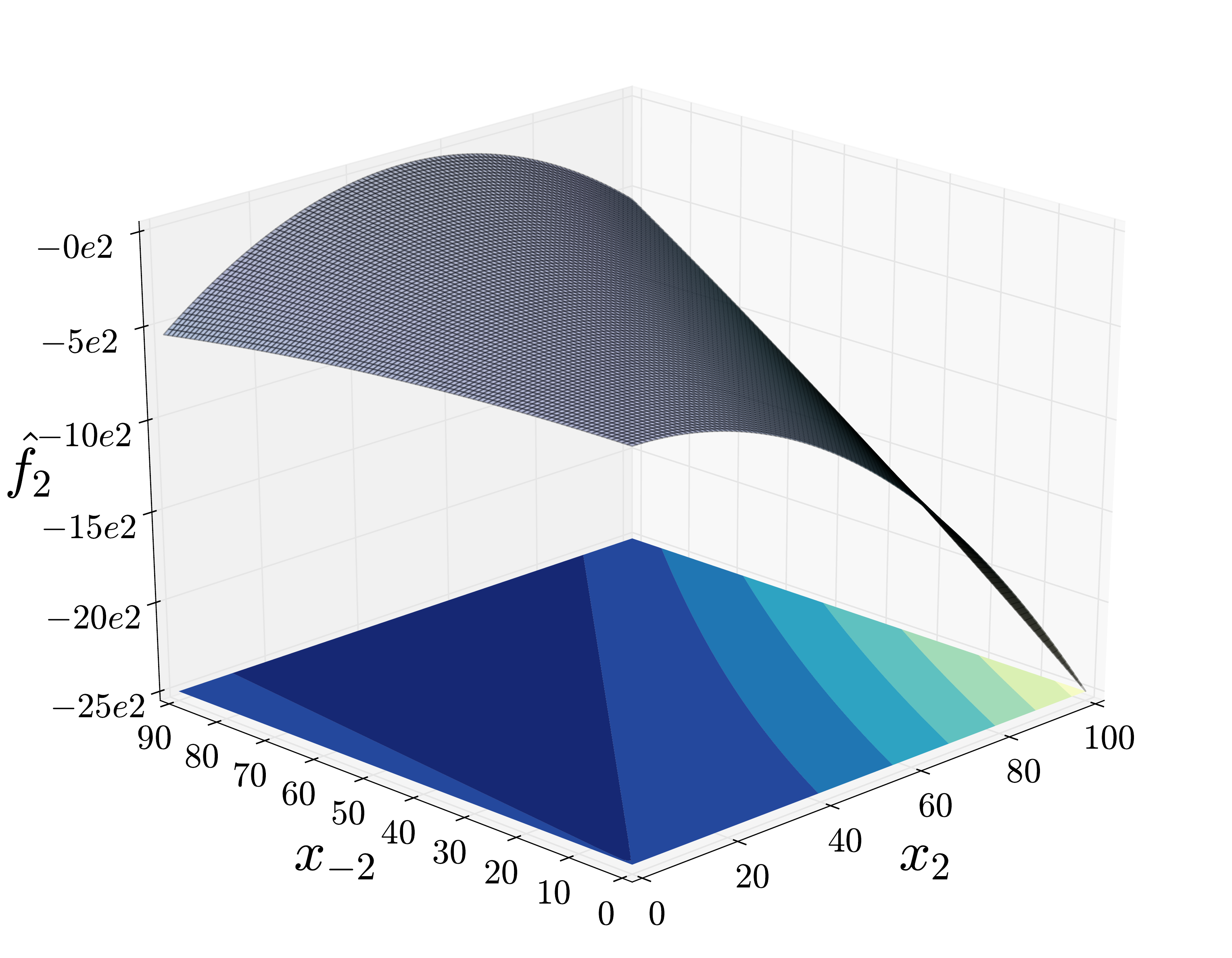}}
    \subfloat[\label{fig:contour_8}]{\includegraphics[width=0.41\textwidth]{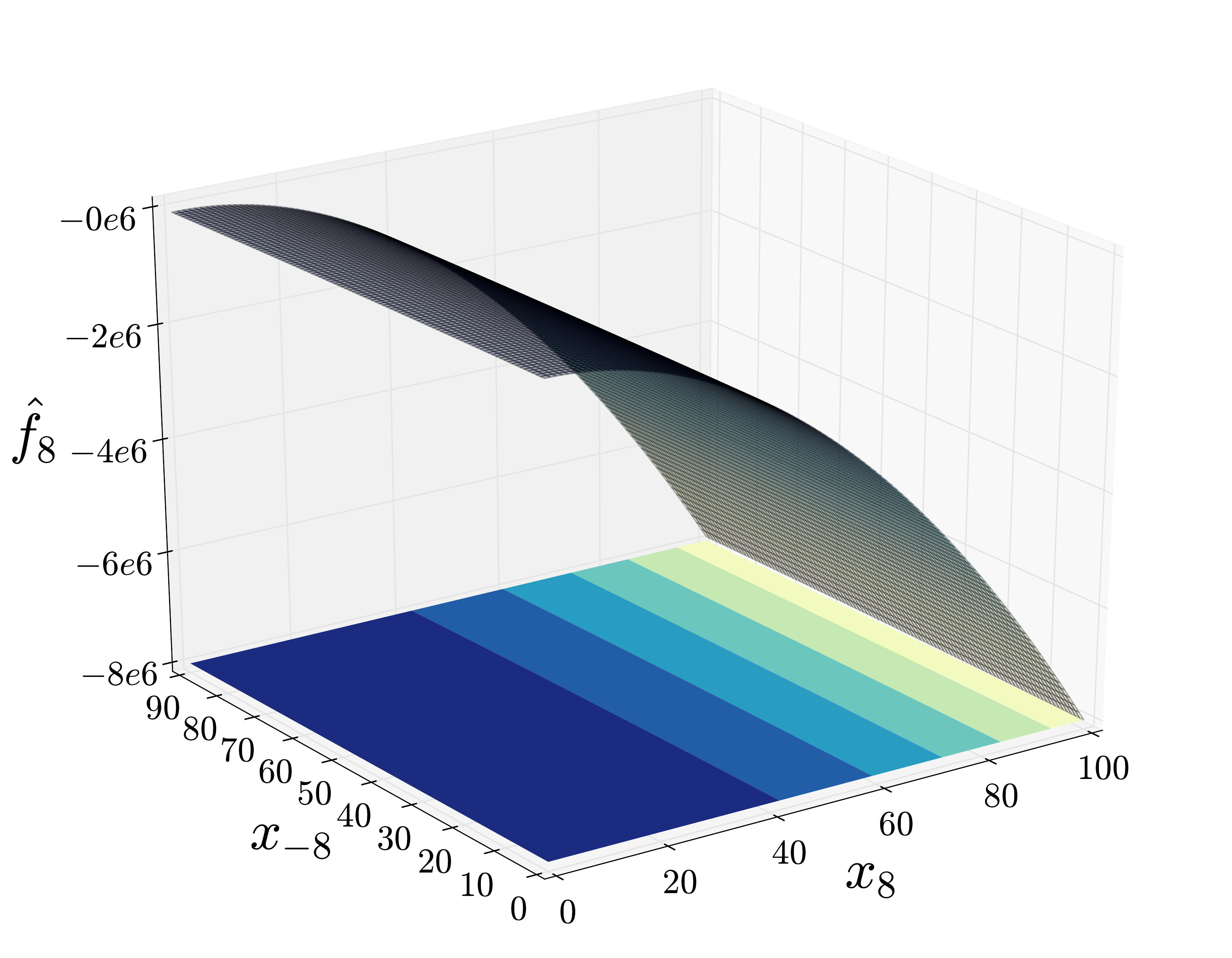}}
  \end{center}
  \caption{Bagging estimated utility functions---using the
  dynamic data set---of (a) agent $2$ and (b) agent $8$. The functions are
  plotted as a function of each agent's own vote $x_2$ (resp.~$x_8$) and other
  players' votes $x_{-2}$ (resp.~$x_{-8}$). Notice that agent $8$, a very
  aggressive player, is indifferent to the choices of the other agents as
  indicated by the fact that its utility is maximized in the same location given
  any value of $x_{-8}$. On the other hand, occupant $2$ responds to changes in
  the other agents' votes and appears to prefer a greater lighting settings
  (more illumination). This indicates that there are different types of players
  and thus, incentives may need to be designed individually for these player
  types in order to elicit the desired response. }
  \label{fig:contour_bagging}
\end{figure*}

\fi

{\begin{figure}[ht]
  \begin{center}
    \includegraphics[width=0.825\columnwidth]{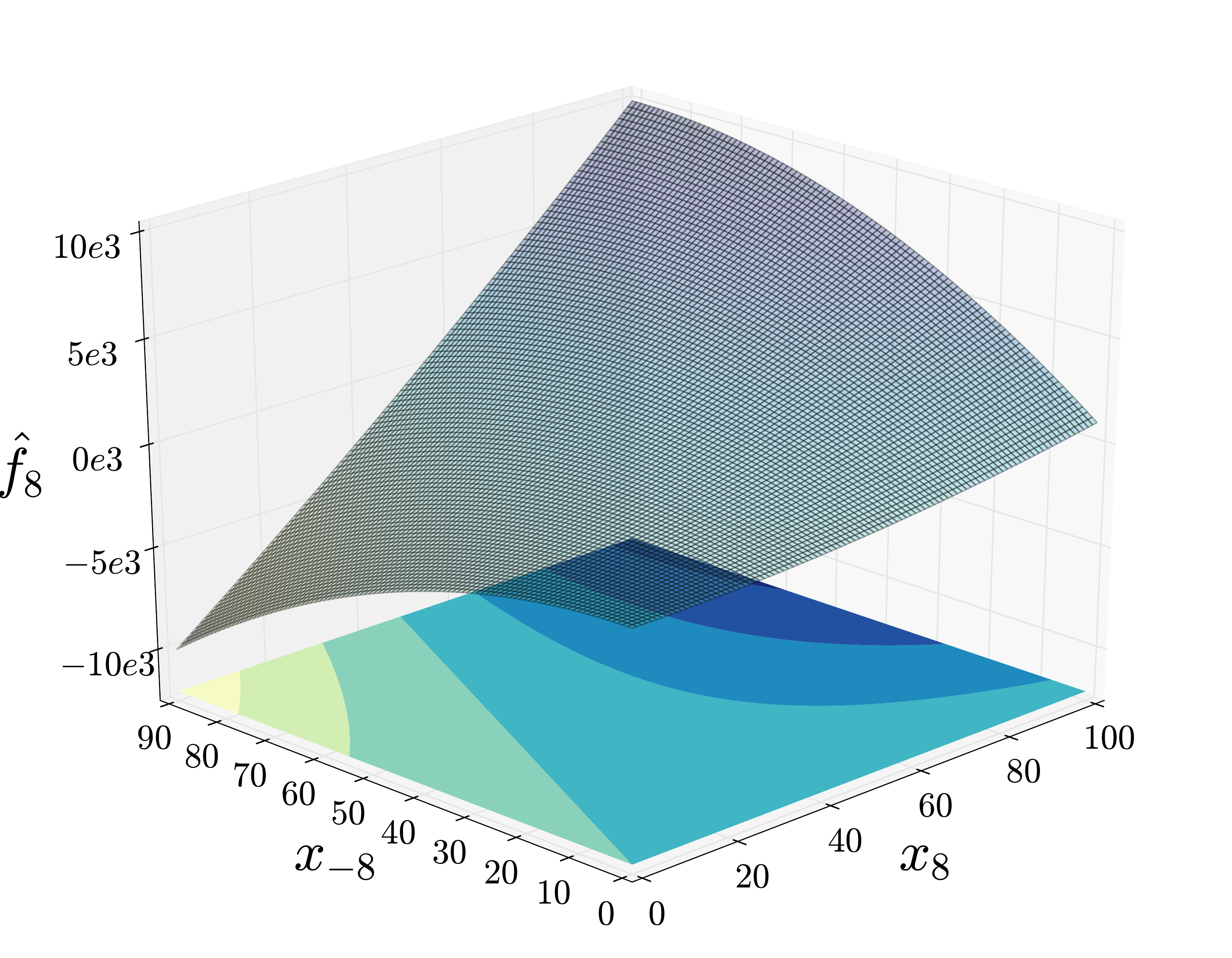}
  \end{center}
  \caption{Agent $8$'s cOLS estimated utility function---using the
  dynamic data set---plotted as a function of $(x_8, x_{-8})$. This figure
  demonstrates that using cOLS (the worst performing estimator) results
in learning a utility function that is not representative of this type of player's
behavior (as can be seen by comparing to Figure~\ref{fig:contour_8}). Incentives
or control designed using this function may result in performance.}
  \label{fig:contour_cOLS}
\end{figure}
}

\subsection{Estimated Utility Functions}
Figure~\ref{fig:contour_bagging} shows the estimated utility functions and their
contour plots for occupants $2$ and $8$---passive and aggressive occupants
respectively---using the parameters obtained via the bagging
ensemble method with the dynamic data set. We remark that we do not
  observe the actual value of agents' utilities; we instead observe only the agents'
  decisions. The purpose of the figures is to show the estimated utility shapes for
  players with significantly different voting profiles (the observable we
  have).
  The particular occupants we selected
represent players that prefer \emph{winning} to lighting satisfaction (occupant $8$) and
players that prefer \emph{lighting satisfaction} to winning (occupant $2$). In particular, occupant $2$'s
estimated utility function appears to be higher at greater lighting settings. 
Exactly the opposite occurs for occupant $8$ whose estimated utility function indicates
that despite changes in the average lighting vote of other players, occupant $8$
aggressively votes for a zero lighting setting which returns the most points.

For comparison---and to highlight the improvement that the robust utility learning
framework offers---in Figure~\ref{fig:contour_cOLS} we show the estimated utility
function for occupant $8$ using cOLS. What we see is a very different
utility function that indicates occupant $8$ cares more about
lighting satisfaction than winning---indicated by the fact that its utility is not
maximized at zero. 
This is misleading since occupant $8$ predominately votes for zero. This is
significant since incentive/control design based on such an erroneous utility
function may lead to very poor performance and occupant dissatisfaction.

\subsection{Bias Approximation and Bias--Variance Tradeoff}
Forecasting accuracy can be enhanced by allowing for a
small amount of bias if it results in a large reduction in variance. 
For a process $Y = X \theta + \epsilon$, 
the Mean Square Error (MSE) characterizes the \emph{bias--variance
tradeoff}:
\begin{align}\label{eq:biasvariancetradoff}
    \textstyle \text{MSE}(x) &= E [(Y - \theta_{\text{est}}^{\top} x)^{2}] \\
   & =\underbrace{(E[\theta_{\text{est}}^{\top} x] - Y)^{2}}_{\text{bias}} + \underbrace{E [
        (\theta_{\text{est}}^{\top} x - E[\theta_{\text{est}}^{\top} x])^{2}] }_{\text{variance}}
\end{align} 
Introducing bias in exchange for reduced variance is 
widely used in ridge regression and in lasso techniques in the form of 
\textit{a priori} knowledge~\cite{Hastie09}. In our robust utility learning framework, we introduce noise structures
that approximate the true data process so that we can fit
cFGLS estimators that are nearly unbiased for those players whose historical
voting record has a large amount of variation.


{
\begin{table}[]
\centering
\caption{{The cFGLS estimator value and the bagging, gradient
boosting and bumping ensemble methods bias approximation for the most active
users. We utilized the dynamic data set from the period in which the default
lighting setting was set to $20$. In bold, we denote the occupants with nearly
unbiased  estimators. }
 }
\label{tab:bias_table}
\begin{tabular}{|c|c|c|c|c|}
\hline
Id & cFGLS & Bagging Bias & Boosting Bias & Bumping Bias \\ \hline\hline
\textbf{2 }      & \textbf{-0.7}  &  \textbf{0.11}  & \textbf{0.17}
&\textbf{0.02} \\ \hline
6      & 0.5        &         1.12        & 1.77        &         0.93 \\ \hline
8       & 298.1        &         -176.9         & -370.3         &         120.5 \\ \hline
14       & 337.5         &         -186.3        & -400.2        &         149.7 \\ \hline
\textbf{20}      & \textbf{-0.8}         &         \textbf{0.07}        &
\textbf{0.21}        &         \textbf{-0.53} \\ \hline
\end{tabular}
\end{table}
}

We approximate the bias for each of the estimators. In
Table~\ref{tab:bias_table}, we present cFGLS estimates obtained using the
dynamic data during the time window in which the default lighting setting was
$20$\%\footnote{The results for the other default lighting settings are
similar.} for selected occupants---the
most active players---as
well as the approximated bias for the estimates generated by bagging,
 bumping, and boosting.

\begin{figure*}[ht!]
  \begin{center}
    \subfloat[\label{fig:unbiased_dynamic}]{\includegraphics[width=0.45\textwidth]{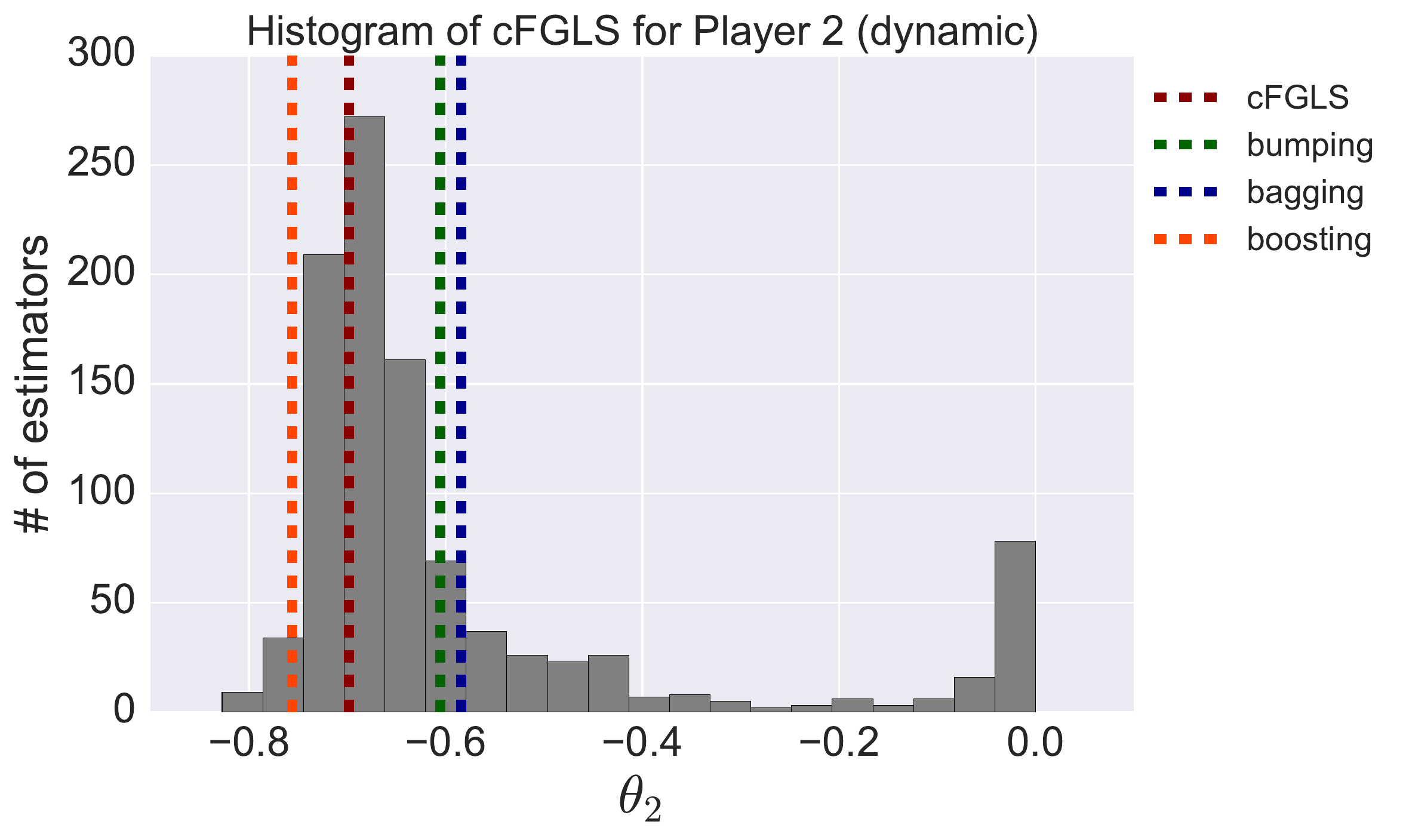}}
    \subfloat[\label{fig:biased_dynamic}]{\includegraphics[width=0.45\textwidth]{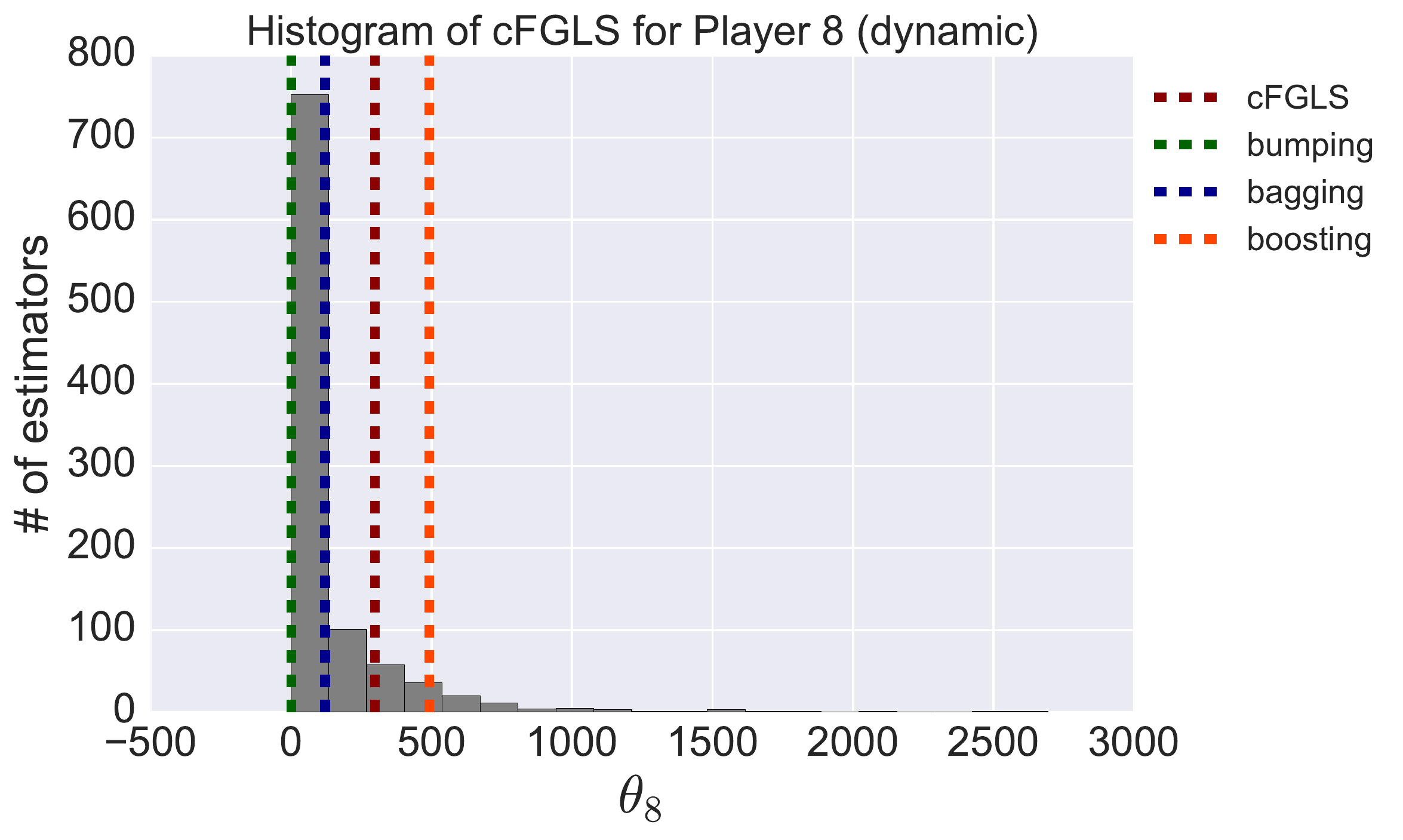}}
  \end{center}
  \caption{ The histograms depict the estimates generated with the wild
  bootstrapping technique using the \emph{dynamic data set} for (a) player $2$ and
  (b) player $8$. The vertical lines mark the value of the
  \textcolor{red!50!black}{\textbf{cFGLS
  (red)}}, \textcolor{dgreen!50!black}{\textbf{bumping (green)}},
  \textcolor{blue!50!black}{\textbf{bagging
  (blue)}}, and \textcolor{orange!85!red!95!black}{\textbf{boosting
  (orange)}} estimators. 
  The histogram for player $2$ is approximately normally distributed around the
  initial cFGLS estimator, indicating that it is unbiased. 
  On the other hand, this is not the case for player $8$. 
  Thus its cFGLS
  estimator is biased. 
  Overall,
  the majority of the proposed ensemble methods result in a significant reduction in
  variance in exchange for an small increase in bias and greater forecasting
  accuracy. In our other work~\cite{konstantakopoulos:2016ab}, we
  develop a hierarchical mixture model that considers both bias and variance.}
  \label{fig:un_bias_dynamic}
\end{figure*}

\begin{figure}[h!]
  \begin{center}
    \includegraphics[width=0.9\columnwidth]{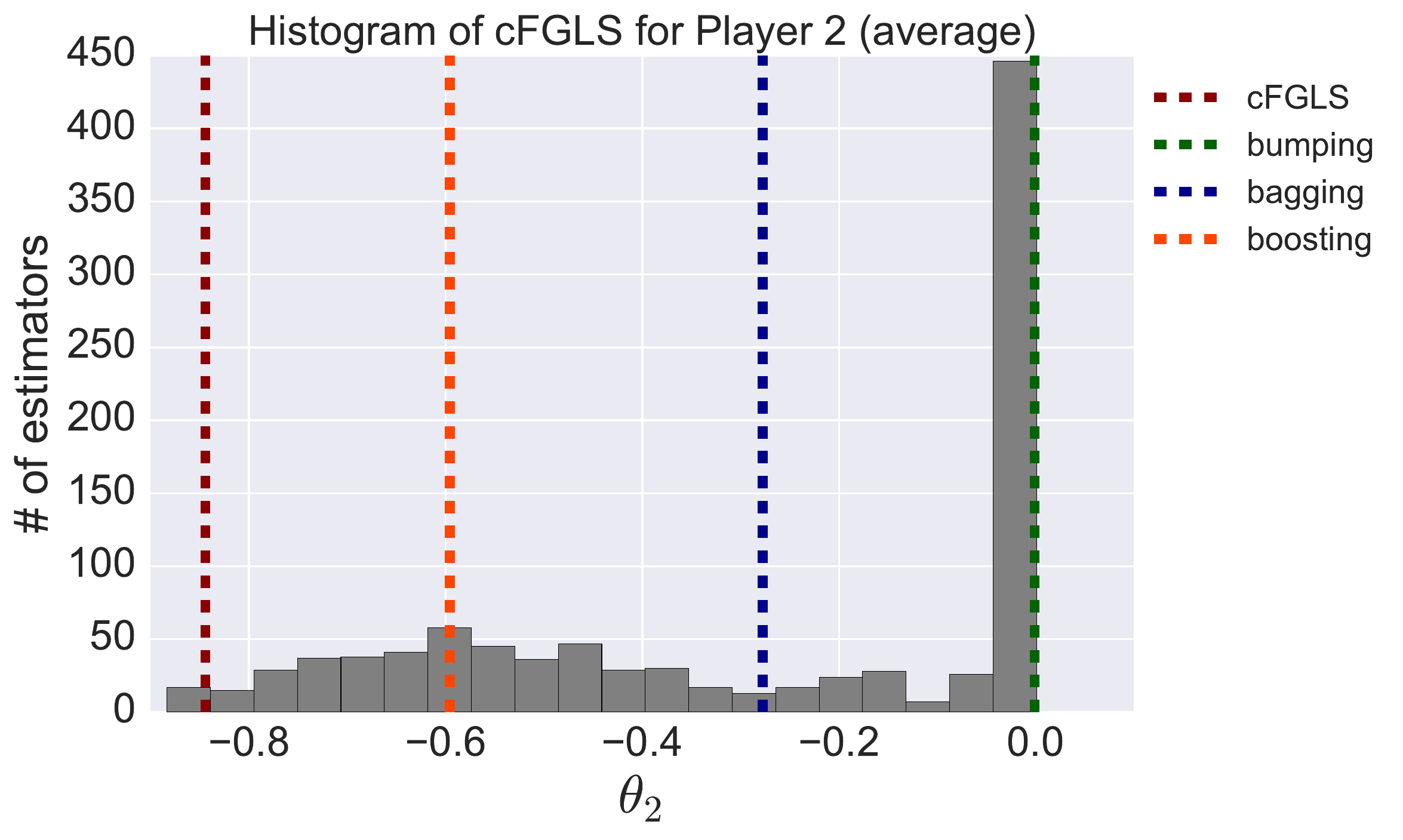} 
  \end{center}
  \caption{The histogram depicts estimator values for player $2$ using the
  wild bootstrapping technique using the \emph{average data set}. The vertical
  lines mark the value of the \textcolor{red!50!black}{\textbf{cFGLS
  (red)}}, \textcolor{green!35!black}{\textbf{bumping (green)}},
  \textcolor{blue!50!black}{\textbf{bagging
  (blue)}}, and \textcolor{orange!85!red!95!black}{\textbf{boosting
  (orange)}} estimators.  We remark that the estimators are all biased.
  This is expected due to limited sample
  size of the average data set. Thus, the average data set cannot be used for optimizing the bias-variance tradeoff.}
  \label{fig:un_bias_average}
\end{figure}

Figures~\ref{fig:un_bias_average} and \ref{fig:un_bias_dynamic} contain
histograms of the cFGLS estimators obtained using the bootstrapped average and
dynamic data,
respectively. In each of these histograms, we also indicate the original
cFGLS\footnote{This is the cFGLS estimator produced using the original average
and dynamic data sets
and not the bootstrapped data sets.} (indicated in
\textcolor{red!50!black}{\textbf{red}}), bagging (indicated in
\textcolor{blue!50!black}{\textbf{blue}}), bumping (indicated in
\textcolor{green!35!black}{\textbf{green}}), and  boosting (indicated in
\textcolor{orange!85!red!95!black}{\textbf{orange}}) estimators with
dashed vertical lines.

The histogram in Figure~\ref{fig:un_bias_average} contains the cFGLS estimators
for occupant $2$. This histogram is representative of the other occupants for
the average data set.
We see that the original cFGLS, bagging, bumping, and boosting estimators each
show some amount of bias. This is largely due to the fact that the average data
set has a small sample size.


On the other hand, in Figure~\ref{fig:unbiased_dynamic} we show the histogram of 
 cFGLS estimators for occupant $2$  produced via bootstrapped dynamic data and
 we can see that
 the original cFGLS estimator (vertical \textcolor{red!50!black}{\textbf{red}}
 line) is nearly unbiased, indicated by the approximate Gaussian distribution around the cFGLS estimate. This is generally true for the occupants with the
 most variation and frequency in their voting record. However, bagging, bumping,
 and boosting
produce estimates that are slightly
biased in exchange for a reduction in estimator variance---see~\eqref{eq:biasvariancetradoff}. 

Occupant $2$ is representative of players which prefer to focus on lighting
satisfaction as opposed to winning whereas
occupant $8$ is representative of players which prefer winning to lighting
satisfaction.
While a very active voter, frequently participating in the
game, occupant $8$'s voting record has little variation (the majority of the time
$x_8=0$).
Figure~\ref{fig:biased_dynamic} contains the cFGLS estimators for occupant $8$
and we see that each of the estimators are slightly biased. Again, these
estimators introduce bias in exchange for a reduction in variance.

\begin{table*}
  \centering
    \caption{Estimated covariance matrix for the most active players using the
    ({\bf a}) dynamic data set and ({\bf b}) average data set. 
 The colored column-row pairs indicate the agents whose utilities we modify to
 create the correlated game; the column indicates the agent(s) whose estimated parameter is
 used to modify the row agent's utility.
 In particular, 
 agent $2$'s utility function is modified by
 by agent $20$'s estimated parameter (\textbf{\textcolor{red!50!black}{red}}),
 agent $8$'s utility function is modified by
 agent $14$'s estimated parameter (\textbf{\textcolor{green!35!black}{green}}),
 and agent $14$'s utility function
     is modified by agent $2$'s and
     agent $8$'s estimated parameter (\textbf{\textcolor{blue!55!black}{blue}}).
     Note that agents $2$ and $14$ are anti-correlated,
     where agents $8$ and $14$ (resp.~agents $2$ and $20$) are positively
     correlated. Agents $2$ and $20$ are passive players, voting more for
     comfort than winning, where agents $8$ and $14$ vote more aggressively. }
\subfloat[][Average Data]{  \begin{tabular}{l|}
    Id\\ \hline\hline
    \textbf{ \textcolor{red!50!black}{2}}\\\hline
    6\\\hline \textbf{\textcolor{green!35!black}{8}}\\\hline
    \textbf{\textcolor{blue!55!black}{14}}\\\hline
    20 \\\hline
  \end{tabular}
  \begin{tabular}{|c|c|c|c|c|}
    \textbf{\textcolor{blue!55!black}{2}} & 6 &\textbf{ \textcolor{blue!55!black}{8}}
    &\textbf{\textcolor{green!35!black}{ 14}} &
    \textbf{\textcolor{red!50!black}{20}}\\
    \hline\hline
    0.086 & 0.080 & -0.190 &
    {-0.248} &\cellcolor{orange!15} \textbf{\textcolor{red!50!black}{0.059}} \\
    \hline
 0.080 & 7.56 & 8.64 & 9.02 & 0.028\\
\hline
-0.190 & 8.64 & 170.98 &
\cellcolor{green!15} \textbf{\textcolor{green!35!black}{44.29}} & -0.337\\
\hline
\cellcolor{blue!15}\textcolor{blue!55!black}{\textbf{-0.248}} & 9.02
&\cellcolor{blue!15}\textcolor{blue!55!black}{\textbf{44.29}} & 87.34  & -0.312\\
\hline
0.059 & 0.028 & -0.337 & -0.312 & 0.063\\
\hline
\end{tabular}}\hspace{1cm}
\subfloat[][Dynamic Data]{\begin{tabular}{l|}
    Id\\ \hline\hline
    \textbf{ \textcolor{red!50!black}{2}}\\\hline
    6\\\hline \textbf{\textcolor{green!35!black}{8}}\\\hline
    \textbf{\textcolor{blue!55!black}{14}}\\\hline
    20 \\\hline
  \end{tabular}
  \begin{tabular}{|c|c|c|c|c|}
    \textbf{\textcolor{blue!55!black}{2}} & 6 &\textbf{ \textcolor{blue!55!black}{8}}
    &\textbf{\textcolor{green!35!black}{ 14}} &
    \textbf{\textcolor{red!50!black}{20}}\\
    \hline\hline
    0.044 & 0.059 & -2.805 &
    {-5.191} &\cellcolor{orange!15} \textbf{\textcolor{red!50!black}{0.031}} \\
    \hline
 0.059 & 7.836 & -16.82 & 0.844 & -0.016\\
\hline
-2.805 & -16.82 & 6.43$\times$10$^4$ &
\cellcolor{green!15} \textbf{\textcolor{green!35!black}{4.28$\times$10$^4$}} & -7.60\\
\hline
\cellcolor{blue!15}\textcolor{blue!55!black}{\textbf{-5.191}} & 0.844
&\cellcolor{blue!15}\textcolor{blue!55!black}{\textbf{{4.28$\times$10$^4$}}} & 8.84$\times$10$^4$  & -12.59\\
\hline
0.031 & -0.016 & -7.60 & -12.59 & 0.073\\
\hline
\end{tabular}}

\label{tab:cov}
\end{table*}


\begin{figure*}[h!]
\center    \subfloat[\label{fig:dynamic-corr}]{\includegraphics[width=0.515\textwidth]{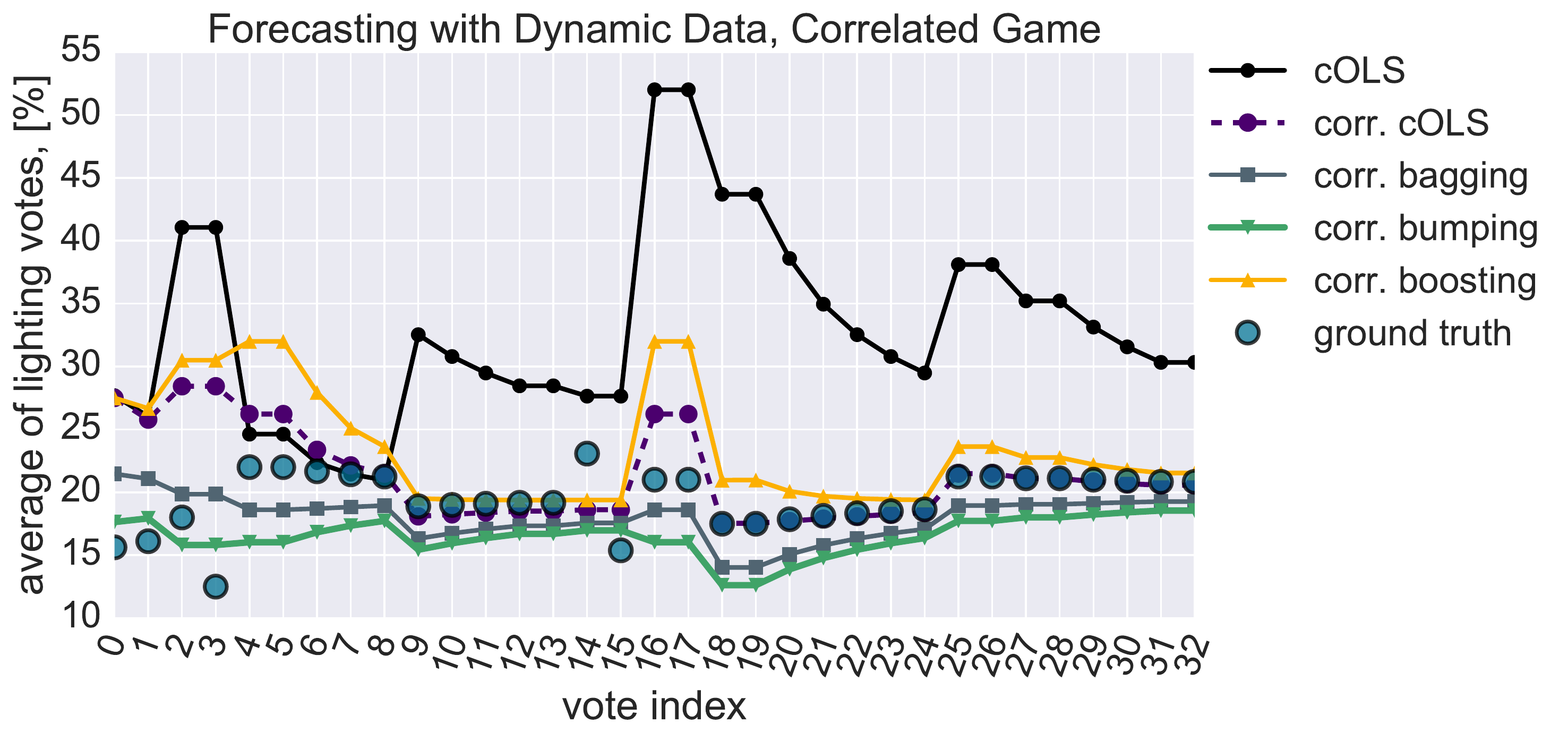}}
    \subfloat[\label{fig:average-corr}]{\includegraphics[width=0.425\textwidth]{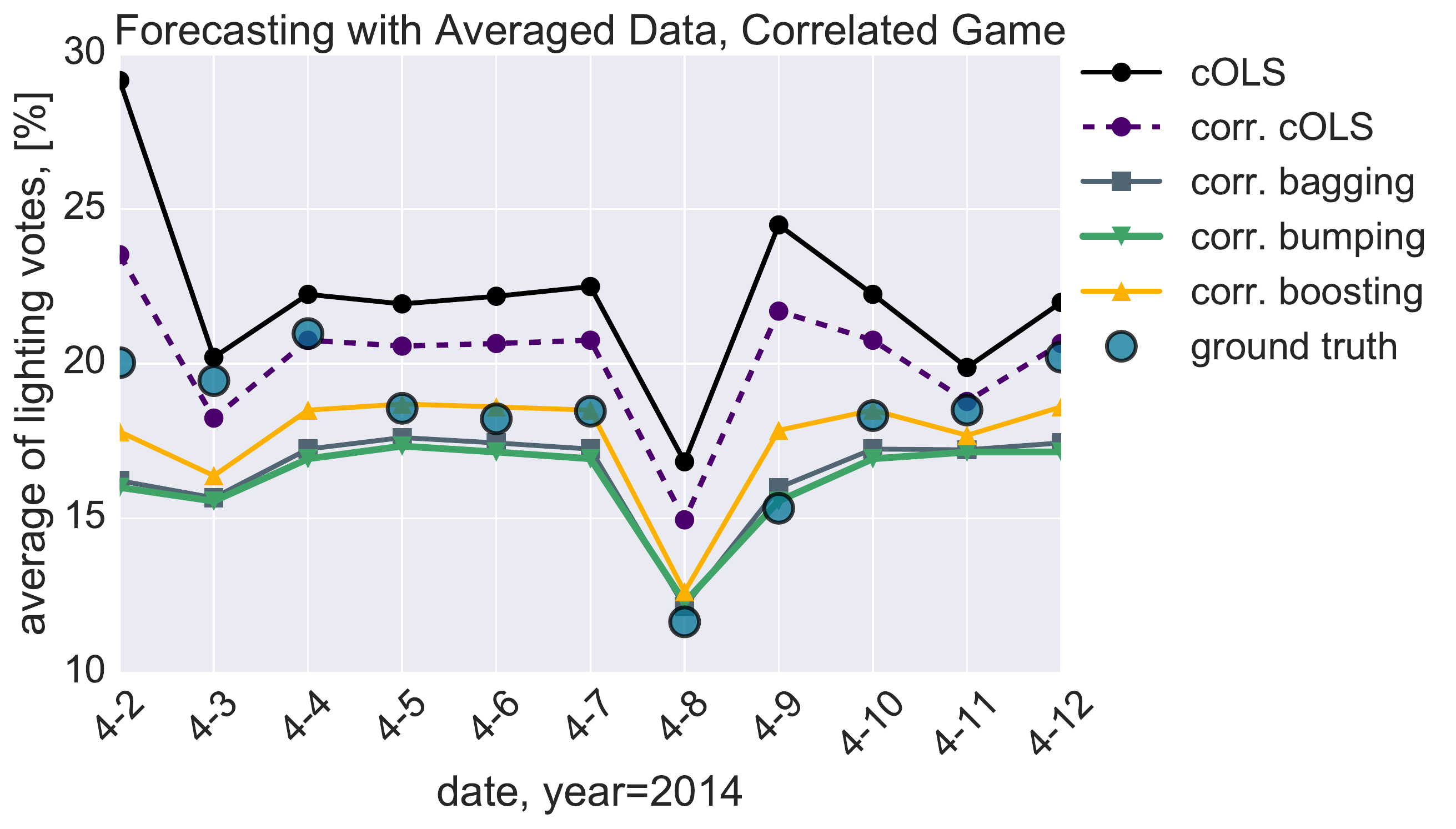}}
  \caption{Forecasting results for the correlated game using (a) dynamic data and (b) averaged
  data for the default lighting setting $20$: For
  the dynamic data, the $x$--axis values indicate the index of when a choice was made by one or more of the occupants (i.e. when
  the implemented lighting setting is changed); the time from one index to the
  next may be several minutes to hours depending on the activity of the
  occupants. For the averaged data, the $x$--axis values are dates (month and
  day).  The ground truth of the average of the lighting votes
  is depicted by the \textbf{\textcolor{blue!50!green!85!black}{blue dots}}; the
  forecast for cOLS is depicted in \textbf{black}; the forecast for correlated
  cOLS is depicted in \textbf{\textcolor{purple!45!black!75!blue}{purple}}; the
  forecast for correlated bagging is depicted in \textbf{\textcolor{gray!85!blue!85!black}{gray}};
  the forecast for correlated bumping is depicted in
  \textbf{\textcolor{green!50!black}{green}}; the
  forecast for correlated boosting is depicted in
  \textbf{\textcolor{yellow!40!orange!90!black}{gold}}.
The forecast for the robust utility learning methods is approximately near the
ground truth
for both data sets while the cOLS estimates produce Nash equilibria with a large
error. However, the correlated cOLS forecast significantly improves on the cOLS
forecast. }
  \label{fig:sims-corr}
\end{figure*}

\subsection{Forecasting via Approximated Correlated Game}
\label{subsec:corr_res}

We now show the results for the correlated utility learning method. Let us use the notation 
\begin{align}
  \hat{g}_i(x_i,x_{-i}; \{\hat{\theta}_j\}_{j\in
\mc{K}_i})=\textstyle\sum_{j\in \mc{K}_i}&
z_{i,j}\sigma_{i,j}\psi_i(x_i,x_{-i})\notag\\
&+\sigma_{i,j}\hat{\theta}_j\phi_i(x_i,x_{-i})  
  \label{eq:hatgres}
\end{align}
where recall that $\mc{K}_i\subset \mc{I}$ is the index set for the players whose
parameters are used to modify player $i$'s utility function in generating the
correlated game and $\hat{\theta}_j$ is the estimated parameter from the
utility learning methods including cOLS, cFGLS, bagging, bumping, and boosting. 
We use the
notation $\hat{g}_i(\cdot; \{\hat{\theta}_j\}_{j\in \mc{K}_i}\})$ as
short-hand.

In Table~\ref{tab:cov}, we show a subset of the estimated covariance matrices obtained using
the dynamic and average data sets. Using these values, we construct the following correlated game. 
Player $2$'s
utility function is modified by player $20$'s:
\begin{align}
  \hat{g}_2(x_2,x_{-2}; \mc{K}_2)& 
  =\textstyle\left(z_{2,2}\sigma_{2,2}+z_{2,20}\sigma_{2,20}\right)\psi_2(x_2,x_{-2})\notag\\
  &+(\sigma_{2,2}\hat{\theta}_2+\sigma_{2,20}\hat{\theta}_{20}
  )\phi_2(x_2,x_{-2})
  \label{eq:2and20}
\end{align}
where $\mc{K}_2=\{2,20\}$. Player $2$ and $20$ are passive players in that their votes tend to be strongly
related to their lighting satisfaction as opposed to increasing their chances of winning. They
are also very active players, having a lot of variation in their voting record.
These two players
are positively correlated with one another (see the
\textcolor{red!50!black}{\bf red} cells in Table~\ref{tab:cov}).

On the other hand, player $8$ and $14$ are aggressive players in that their
votes tend to be much lower indicating a greater desire to win points. These
players are also positively correlated (see the \textcolor{green!35!black}{\bf green}
cell's in Table~\ref{tab:cov}).
With this in mind, we modify player $8$'s utility function by player $14$'s:
\begin{align}
  \hat{g}_8(x_8,x_{-8};\mc{K}_8)& 
  =\left(z_{8,8}\sigma_{8,8}+z_{8,14}\sigma_{8,14}\right)\psi_8(x_8,x_{-8})\notag\\
  &+( \sigma_{8,8}\hat{\theta}_8+\sigma_{8,14}\hat{\theta}_{20}
  )\phi_8(x_8,x_{-8})
  \label{eq:8and14}
\end{align}
where $\mc{K}_8=\{8,14\}$. 

Player $14$ is also negatively correlated with player $2$. Hence, player $14$'s
utility function is modified by player $2$'s and $8$'s utilities. That is, with
$\mc{K}_{14}=\{2,8,14\}$, we have
\begin{align}
  \hat{g}_{14}(x_{14},x_{-14}; \mc{K}_{14})
  &=\textstyle\sum_{i\in\mc{K}_{14}}(z_{14,i}\sigma_{14,i}
  \psi_{14}(x_{14},x_{-14})\notag\\
  &\quad+\sigma_{14,i}\hat{\theta}_{i}\phi_{14}(x_{14},x_{-14}))
%
  \label{eq:play14}
\end{align}

{\begin{table}[]
\centering
\caption{\scriptsize Root Mean Square Error (RMSE), Mean Absolute Error (MAE)
and Mean Absolute Scaled Error (MASE) of forecasting using the estimated
correlated utility functions. We estimated correlated utility functions
$\hat{g}_i(\cdot;\{\theta_j\}_{j\in \mc{K}_i})$ using parameters from the
 bagging, bumping, boosting, and cOLS methods 
for both data sets in default lighting setting $20$. }
 \label{tab:rmse_20-corr}
\begin{tabular}{|l|c|c|c|c|}
\hline
{\small \emph{\textbf{Dynamic, $\hat{g}_i$}}} & bagging & boosting
& bumping &cOLS  \\ \hline
{  RMSE }& \textbf{6.38} & 9.58 & 8.82 & 8.44 \\ \hline
\textit{ MAE }& \textbf{4.59} & 6.81 & 5.52 & 5.58 \\ \hline
\textit{ MASE }&\textbf{ 1.84 }& 2.72 & 2.21 & 2.23  \\ \hline \hline
{\small \emph{\textbf{Averaged, $\hat{g}_i$}}} &  bagging &  boosting &
bumping &
 cOLS  \\ \hline
 \textit{ RMSE }& 2.18 & \textbf{1.63} & 2.36  & 2.83 \\ \hline
 \textit{ MAE }& 1.75 & \textbf{1.27} & 1.92 & 2.30 \\ \hline
 \textit{ MASE }& 0.78 & \textbf{0.56} & 0.86 & 1.03  \\ \hline
\end{tabular}
\end{table}
}


All the other players'
utilities in the correlated game remain unchanged; that is, they are taken to be $\hat{g}_i=\hat{f}_i$, $i\in\mc{I}/\{2,8,14\}$.

These player combinations were selected since, through the
correlated game,
we aim to {improve} our estimators by leveraging correlations between
players. In particular, the goal is to utilize information learned from players
with the most variation in their votes in improving the estimates of
players who consistently vote the same value or have a limited
participation record.

In Table~\ref{tab:rmse_20-corr}, we
present the RMSE, MAE, and MASE for the estimated correlated game 
$\{\hat{g}_i(\cdot;\{\hat{\theta}_j\}_{j\in \mc{K}_i})\}_{i\in \mc{I}}$
where the $\hat{\theta}_j$'s are taken to be the cOLS, bagging, boosting, and
bumping estimators. Comparing these results to those in Table~\ref{tab:rmse_20},
we see that correlated estimation schemes applied to the dynamic data set reduce the estimation error for almost
every method. Moreover, correlated bagging outperforms bagging, the best performing
ensemble method, by all three metrics. For the average data set, correlated
boosting outperforms the best performing ensemble method, boosting, again by
all three metrics. 

In Figure~\ref{fig:sims-corr}, we show the forecast produced by the correlated
utility learning method using the cOLS, bagging, bumping, and boosting
estimators  and 
the ground truth test data. Figure~\ref{fig:dynamic-corr}
and~\ref{fig:average-corr} are the forecasts for the dynamic and average data
sets, respectively.


What is perhaps most interesting is that, for both data sets, the correlated
cOLS results improve the forecasting error as compared to cOLS and
the results are not significantly different than the
other ensemble methods. This can be seen in Table~\ref{tab:rmse_20-corr} and
Figure~\ref{fig:sims-corr}. 
The importance of this finding is that correlated cOLS has the potential to be
integrated into an online algorithm.
The classical cOLS can be
performed online and is, thus, amenable to an online incentive design
framework~\cite{ratliff:2015aa, ratliff:2014aa}. However, as we have seen, the ensemble methods
significantly outperform cOLS. Determining the estimated
covariance matrix requires solving a generalized least squares (GLS) and noise
covariance estimation problem~\cite{isermann:2011aa}. 
Given that the estimated correlated game using
cOLS parameters provides nearly the same estimation error as the ensemble
methods, these methods can be adapted to estimate the correlated game parameters
and then introduced into an adaptive incentive design framework. We are
currently exploring this extension as the ultimate objective is to utilize the
learned utilities in an incentive design framework, preferably one that can be
executed in an adaptive/online manner. This will support a more
robust online utility learning and incentive design algorithm.


\section{Discussion}
\label{sec:conclusion}
We presented a general framework for robust utility learning using a
heteroskedastic inference adaptation to cGLS and we leveraged learned
correlations between players in constructing a correlated utility learning
framework that matches the robust utility learning errors while also being
amenable to online implementation. The latter is important for integrating the
proposed utility learning techniques with adaptive control or online incentive
design. For example, it has been shown that static programs for encouraging
energy efficiency are subject to the rebound effect in which participants often
return to less efficient behavior after some time~\cite{laitner:2000aa,schipper:2000aa}. By integrating our utility
learning framework with incentive design, we will be able to create an adaptive
model that learns how users' preferences change over time and thus, generate the
appropriate incentives to ensure active participation. 

To demonstrate the utility learning methods, we applied them to data collected
from a smart building social game we conducted where occupants vote for shared resources and
participate in a lottery.
We were able to estimate nearly unbiased estimators for several agent
profiles and significantly reduce the forecasting error as compared
to cOLS.
The robust utility learning framework enables us to effectively
\emph{close the loop} around smart building occupants by providing the
foundation for learning a decision-making model that can be integrated into the
incentive or control design process. While we apply the method to smart
building social game data, it can be applied more generally to scenarios with the task of inverse modeling of competitive agents
and provides a useful tool for many
 smart infrastructure applications where learning decision--making
behavior is crucial.

\section*{Acknowledgment}
We thank Mr. Christopher Hsu, Applications Programmer at CREST laboratory, who
developed and deployed the web portal application of the social game at UC Berkeley. 
\appendix
\section{Proof of Proposition~\ref{prop:suffcond}}
\label{app:proofs}
\begin{proof}[Proof of Proposition~\ref{prop:suffcond}]
 Suppose the assumptions hold. The constraints for each player do not depend
   on other players' choice variables. We can hold $x_{-i}^\ast$ fixed and apply
   Proposition 3.3.2~\cite{Bertsekas:1999fk} to the $i$-th player's optimization
   problem $\max\left\{ f_i(x_i, x_{-i}^\ast)\ | \  x_i\in \mc{C}_i\right\}$.
   Since each $f_i$ is concave and each $\mc{C}_i$ is a convex set, $x_i^\ast$ is a
   global optimum of the $i$-th player's optimization problem under the
   assumptions. Since this is true for each of the $i\in\{1,\ldots, n\}$
   players, $x^\ast$ is a Nash equilibrium.
 \end{proof}

\bibliographystyle{IEEEtran}
\bibliography{journal_utility}


%
%
%
%
%
%
%
%
%
%

\end{document}